\documentclass[11pt]{article}

\usepackage[utf8]{inputenc}
\usepackage[margin=1in]{geometry}

\usepackage{amsmath}                           
\usepackage{amssymb}
\usepackage{amsthm}

\usepackage{bm}
\usepackage{verbatim}

\usepackage{wrapfig}

\newtheorem{theorem}{Theorem}
\newtheorem{lemma}[theorem]{Lemma}

\newtheorem{definition}[theorem]{Definition}

\newtheorem{corollary}[theorem]{Corollary}

\newtheorem{observation}[theorem]{Observation}

\usepackage{url}

\usepackage[table]{xcolor}
\usepackage{array}

\usepackage[inline]{enumitem}
\usepackage{moreenum}

\usepackage{mathtools}
\usepackage[super]{nth}

\setlist[enumerate]{nosep,topsep=0em}
\setlist[enumerate,1]{label=(\roman*), leftmargin=2.2em}
\setlist[enumerate,2]{label=(\alph*)}
\setlist[itemize]{nosep,topsep=0.1em}

\usepackage{xfrac}

\usepackage{tikz}
\usetikzlibrary{positioning}
\usetikzlibrary{calc}
\usetikzlibrary{shapes.geometric}
\usetikzlibrary{arrows.meta}
\usetikzlibrary{decorations.pathreplacing, decorations.pathmorphing}

\usetikzlibrary{matrix, calc, shapes}
\tikzset{
  treenode/.style = {shape=rectangle, rounded corners,
                     draw, anchor=center,
                     text width=6em, 
                     align=center,
                     top color=white, 
                     inner sep=0.5ex},
  decision/.style = {treenode, diamond, inner sep=-0.5pt,bottom color=blue!10},
  root/.style     = {treenode, bottom color=red!20},
  env/.style      = {treenode, bottom color=green!10},
  dummy/.style    = {circle,draw}
}

\usepackage{tcolorbox}
\tcbuselibrary{skins}

\usepackage{caption}

\usepackage[vlined,ruled,algo2e]{algorithm2e}
\SetAlCapHSkip{0.2em}

\usepackage{mdframed}

\usepackage[pdfpagelabels,bookmarks=false]{hyperref}
\hypersetup{colorlinks, linkcolor=darkblue, citecolor=darkgreen, urlcolor=darkblue}

\makeatletter
\newcommand{\labeltarget}[1]{\Hy@raisedlink{\hypertarget{#1}{}}}
\makeatother

\usepackage{etoolbox}

\newcommand{\repthanks}[1]{\textsuperscript{\ref{#1}}}

\makeatletter
\patchcmd{\maketitle}
  {\def\thanks}
  {\let\repthanks\repthanksunskip\def\thanks}
  {}{}
\patchcmd{\@maketitle}
  {\def\thanks}
  {\let\repthanks\@gobble\def\thanks}
  {}{}
\newcommand\repthanksunskip[1]{\unskip{}}
\makeatother

\makeatletter
\def\@fnsymbol#1{\ensuremath{\ifcase#1\or {\star} \or {\star} {\star} \or {\star} {\star} {\star} \or
 \else\@ctrerr\fi}}
\makeatother

\definecolor{darkblue}{rgb}{0,0,0.38}
\definecolor{darkred}{rgb}{0.6,0,0}
\definecolor{darkgreen}{rgb}{0.1,0.35,0}

\DeclareMathOperator{\conv}{conv}

\DeclareMathOperator{\poly}{poly}
\DeclareMathOperator{\VC}{VC}
\DeclareMathOperator{\VCthree}{VC3}

\renewcommand{\P}{\mathtt{P}}
\newcommand{\NP}{\mathtt{NP}}

\newcommand{\g}{\gamma}

\newcommand\gckc{{\ensuremath{\operatorname{\gamma C k C}}}\xspace}
\newcommand\gckcB{{\ensuremath{\operatorname{\mathbf{\bm{\gamma} C k C}}}}\xspace}

\newcommand{\p}{\mathcal{P}} %
\newcommand{\ip}{\mathcal{P}_{I}} %

\newcommand{\PLP}{\textrm{PLP}}
\newcommand{\DLP}{\textrm{DLP}}

\newcommand{\fp}{\p^{\alpha,\mu}}
\newcommand{\fip}{\ip^{\alpha,\mu}}

\newcommand{\q}{\mathcal{Q}}

\DeclareMathOperator*{\argmax}{\arg\!\max}

\begin{document}

\title{A Technique for Obtaining True Approximations for $k$-Center with Covering Constraints \thanks{This project received funding from the European Research Council (ERC) under the European Union's Horizon 2020 research and innovation programme (grant agreement No 817750) and the Swiss National Science Foundation grants 200021\_184622 and PZ00P2$\_$174117.\newline
\indent A preliminary version of this work was presented at the 21st Conference on Integer Programming and Combinatorial Optimization (IPCO 2020). An independent work of Jia, Sheth, and Svensson~\cite{DBLP:conf/ipco/JiaSS20}, presented at the same venue, gave a $3$-approximation for Colorful $k$-Center with constantly many colors using different techniques.
}
}

\author{Georg Anegg         \and
        Haris Angelidakis \and Adam Kurpisz \and
        Rico Zenklusen
}

\date{}

\maketitle

\begin{abstract}
There has been a recent surge of interest in incorporating fairness aspects into classical clustering problems. Two recently introduced variants of the $k$-Center problem in this spirit are Colorful $k$-Center, introduced by Bandyapadhyay, Inamdar, Pai, and Varadarajan, and lottery models, such as the Fair Robust $k$-Center problem introduced by Harris, Pensyl, Srinivasan, and Trinh. To address fairness aspects, these models, compared to traditional $k$-Center, include additional covering constraints. Prior approximation results for these models require to relax some of the normally hard constraints, like the number of centers to be opened or the involved covering constraints, and therefore, only obtain constant-factor pseudo-approximations. In this paper, we introduce a new approach to deal with such covering constraints that leads to (true) approximations, including a $4$-approximation for Colorful $k$-Center with constantly many colors---settling an open question raised by Bandyapadhyay, Inamdar, Pai, and Varadarajan---and a $4$-approximation for Fair Robust $k$-Center, for which the existence of a (true) constant-factor approximation was also open.

We complement our results by showing that if one allows an unbounded number of colors, then Colorful $k$-Center admits no approximation algorithm with finite approximation guarantee, assuming that $\P \neq \NP$. Moreover, under the Exponential Time Hypothesis, the problem is inapproximable if the number of colors grows faster than logarithmic in the size of the ground set.
\end{abstract}

\section{Introduction}\label{sec:intro}

Along with $k$-Median and $k$-Means, $k$-Center is one of the most fundamental and heavily studied clustering problems. In $k$-Center, we are given a finite metric space $(X,d)$ and an integer $k\in [|X|]\coloneqq \{1,\dots, |X|\}$, and the task is to find a set $C\subseteq X$ with $|C|\leq k$ minimizing the maximum distance of any point in $X$ to its closest point in $C$.
Equivalently, the problem can be phrased as covering $X$ with $k$ balls of radius as small as possible, i.e., finding the smallest radius $r\in \mathbb{R}_{\geq 0}$ together with a set $C\subseteq X$ with $|C|\leq k$ such that $X = B(C,r) \coloneqq \bigcup_{c\in C}B(c,r)$, where $B(c,r)\coloneqq \{u\in X: d(c,u)\leq r\}$ is the ball of radius $r$ around $c$.

$k$-Center, like most clustering problems, is computationally hard; actually it is $\NP$-hard to approximate to within any constant below $2$~\cite{DBLP:journals/dam/HsuN79}. On the positive side, various $2$-approximations~\cite{DBLP:journals/tcs/Gonzalez85,DBLP:journals/mor/HochbaumS85} have been found, and thus, its approximability is settled.
Many variations of $k$-Center have been studied, most of which are based on generalizations along one of the following two main axes:
\begin{enumerate}
\item\label{item:varConstrCenters} which sets of centers can be selected, and
\item\label{item:varOutliers} which sets of points of $X$ need to be covered.
\end{enumerate}
The most prominent variations along~\ref{item:varConstrCenters} are variations where the set of centers is required to be in some down-closed family $\mathcal{F}\subseteq 2^X$. For example, if centers have non-negative opening costs and there is a global budget for opening centers, Knapsack Center is obtained. If $\mathcal{F}$ is the set of independent sets of a matroid, the problem is known as Matroid Center. 
The best-known problem type linked to~\ref{item:varOutliers} is Robust $k$-Center. Here, an integer $m\in [|X|]$ is given, and one only needs to cover any $m$ points of $X$ with $k$ balls of radius as small as possible. Research on $k$-Center variants along one or both of these axes has been very active and fruitful, see, e.g.,~\cite{DBLP:journals/jacm/HochbaumS86,DBLP:conf/soda/CharikarKMN01,DBLP:journals/algorithmica/ChenLLW16,DBLP:conf/icalp/ChakrabartyGK16}. In particular, recent work of Chakrabarty and Negahbani~\cite{DBLP:journals/talg/ChakrabartyN19} presents an elegant and unifying framework for designing best possible approximation algorithms for all above-mentioned variants.

All the above variants have in common that there is a single covering requirement; either all of $X$ needs to be covered or a subset of it. Moreover, they come with different kinds of packing constraints on the centers to be opened as in Knapsack or Matroid Center.
However, the desire to address fairness in clustering, which has received significant attention recently, naturally leads to multiple covering constraints. Here, existing techniques only lead to constant-factor pseudo-approximations that violate at least one constraint, like the number of centers to be opened. In this work, we present techniques for obtaining (true) approximations for two recent fairness-inspired generalizations of $k$-Center along axis~\ref{item:varOutliers}, namely 
\begin{enumerate}[topsep=0.3em]
\item $\gamma$-Colorful $k$-Center, as introduced by Bandyapadhyay et al.~\cite{DBLP:conf/esa/Bandyapadhyay0P19}, and
\item Fair Robust $k$-Center, a lottery model introduced by Harris et al.~\cite{HarrisPST19}.
\end{enumerate}

\smallskip

\emph{$\gamma$-Colorful $k$-Center} (\gckc) is a fairness-inspired $k$-Center model imposing covering constraints on subgroups. It is formally defined as follows.\footnote{
The version introduced in~\cite{DBLP:conf/esa/Bandyapadhyay0P19} requires $X_1,\ldots, X_\gamma$ to partition $X$. However, \gckc readily reduces to the more restrictive model in~\cite{DBLP:conf/esa/Bandyapadhyay0P19} by introducing a new color $X_{\gamma + 1} = X \setminus \bigcup_{i \in [\gamma]}X_i$ with $m_{\gamma + 1} = 0$ and replacing each element that has $q > 1$ colors by $q$ elements on the same location with each having a single color.
}

\begin{definition}[$\gamma$-Colorful $k$-Center (\gckc)~\cite{DBLP:conf/esa/Bandyapadhyay0P19}]
Let $\gamma,k\in \mathbb{Z}_{\geq 1}$, $(X,d)$ be a finite metric space, $X_\ell \subseteq X$ for $\ell\in [\gamma]$, and $m \in \mathbb{Z}_{\geq 0}^{\gamma}$. The \emph{$\g$-Colorful $k$-Center problem (\gckc)} asks to find the smallest radius $r\in \mathbb{R}_{\geq 0}$ together with centers $C\subseteq X$, $|C|\leq k$, such that
\begin{equation*}
|B(C,r)\cap X_\ell| \geq m_\ell \quad \forall \ell\in [\gamma]\enspace.
\end{equation*}
Such a set of centers $C$ is called a \gckc solution of radius $r$.
\end{definition}

The name stems from interpreting each set $X_\ell$ for $\ell\in [\gamma]$ as a color assigned to the elements of $X_{\ell}$. In particular, an element can have multiple colors or no color. In words, the task is to open $k$ balls of smallest possible radius such that, for each color $\ell\in [\gamma]$, at least $m_\ell$ points of color $\ell$ are covered. Hence, for $\gamma=1$, we recover the Robust $k$-Center problem.

We briefly contrast \gckc with related fairness models. A related class of models that has received significant attention also assumes that the ground set is colored, but requires that each cluster contains approximately the same number of points from each color. Such variants have been considered for $k$-Median, $k$-Means, and $k$-Center, e.g., see~\cite{DBLP:conf/nips/Chierichetti0LV17,DBLP:conf/icalp/Rosner018,DBLP:conf/approx/Bercea0KKRS019,DBLP:conf/icml/BackursIOSVW19,DBLP:conf/nips/BeraCFN19} and references therein. \gckc differentiates itself from the above notion of fairness by not requiring a per-cluster guarantee, but a global fairness guarantee. More precisely, each color can be thought of as representing a certain group of people (demographic), and a global covering requirement is given per demographic. Also notice the difference with the well-known Robust $k$-Center problem, where a feasible solution might, potentially, completely ignore a certain subgroup, resulting in a heavily unfair treatment. \gckc addresses this issue.

\medskip

The presence of multiple covering constraints in \gckc, imposed by the colors, hinders the use of classical $k$-Center clustering techniques, which, as mentioned above, have mostly been developed for packing constraints on the centers to be opened. An elegant first step was done by Bandyapadhyay et al.~\cite{DBLP:conf/esa/Bandyapadhyay0P19}. They exploit sparsity of a well-chosen LP (in a similar spirit as in~\cite{HarrisPST19}) to obtain the following pseudo-approximation for $\gckc$: they efficiently compute a solution of twice the optimal radius by opening at most $k+\gamma-1$ centers. Hence, up to $\gamma-1$ more centers than allowed may have to be opened.
Moreover,~\cite{DBLP:conf/esa/Bandyapadhyay0P19} shows that in the Euclidean plane, a significantly more involved extension of this technique allows for obtaining a true  $(17+\varepsilon)$-approximation for $\gamma=O(1)$. Unfortunately, this approach is heavily problem-tailored and does not even extend to $3$-dimensional Euclidean spaces. This naturally leads to the main open question raised in~\cite{DBLP:conf/esa/Bandyapadhyay0P19}:
\begin{center}
\emph{Does $\gckc$ with $\gamma=O(1)$ admit an $O(1)$-approximation, for any finite metric?}
\end{center}
Here, we introduce a new approach that answers this question affirmatively.

\medskip

Together with additional ingredients, our approach also applies to Fair Robust $k$-Center, which is a natural lottery model introduced by Harris et al.~\cite{HarrisPST19}. We introduce the following generalization thereof that can be handled with our techniques, which we name \emph{Fair $\gamma$-Colorful $k$-Center problem (Fair~\gckc)}. (The Fair Robust $k$-Center problem, as introduced in~\cite{HarrisPST19}, corresponds to $\gamma=1$.)
\begin{definition}[Fair $\gamma$-Colorful $k$-Center (Fair~$\gckc$)]\label{def:fair-colorful}
Given is a \gckc instance on a finite metric space $(X,d)$ together with a vector $p\in [0,1]^X$. The goal is to find the smallest radius $r\in\mathbb{R}_{\geq 0}$, for which there exists a distribution $\mathcal{H}$ over feasible $\gckc$ solutions of radius $r$ such that
\begin{equation*}
\Pr_{C\sim\mathcal{H}} [u \in B(C,r)] \geq p(u) \quad \forall u\in X\enspace.
\end{equation*}
An algorithm for this problem should return a radius $r$ along with an efficient procedure for sampling a random feasible $\gckc$ solution of radius $r$.
\end{definition}
We note that if there exists a distribution $\mathcal{H}$ with the desired properties for some radius $r$, then there exists a distribution of polynomial support with the desired properties (due to sparsity of the natural LP corresponding to the distribution, described in Section~\ref{sec:lottery}). This, in particular, implies that the corresponding decision problem is in $\NP$.

Fair~\gckc is a generalization of \gckc, where each element $u\in X$ needs to be covered with a prescribed probability $p(u)$. The Fair Robust $k$-Center problem, i.e., Fair~\gckc with $\gamma=1$, is indeed a fairness-inspired generalization of Robust $k$-Center, since Robust $k$-Center is obtained by setting $p(u)=0$ for $u\in X$. One example setting where the additional fairness aspect of Fair \gckc compared to \gckc is nicely illustrated, is when $k$-Center problems have to be solved repeatedly on the same metric space. The introduction of the probability requirements $p$ allows for obtaining a distribution to draw from that needs to consider all elements of $X$ (as prescribed by $p$), whereas classical Robust $k$-Center likely ignores a group of badly-placed elements. We refer to Harris et al.~\cite{HarrisPST19} for further motivation of the problem setting. They also discuss the Knapsack and Matroid Center problem under the same notion of fairness.

For Fair Robust $k$-Center, \cite{HarrisPST19}~presents a $2$-pseudo-approximation that slightly violates both the number of points to be covered and the probability of covering each point. More precisely, for any constant $\varepsilon >0$, only a $(1-\varepsilon)$-fraction of the required number of elements are covered, and element $u\in X$ is covered only with probability $(1-\varepsilon) p(u)$ instead of $p(u)$.
It was left open in~\cite{HarrisPST19} whether a true approximation may exist for Fair Robust $k$-Center.

\subsection{Our results} 

Our main contribution is a method to obtain $4$-approximations for variants of $k$-Center with unary encoded covering constraints on the points to be covered. We illustrate our technique in the context of $\gckc$, affirmatively resolving the open question of Bandyapadhyay et al.~\cite{DBLP:conf/esa/Bandyapadhyay0P19} about the existence of an $O(1)$-approximation for constantly many colors (without restrictions on the underlying metric space).
\begin{theorem}\label{thm:mainGckc}
There is a $4$-approximation algorithm for $\gckc$ running in time $|X|^{O(\gamma)}$.
\end{theorem}

\noindent \textbf{Note:} In an independent work, Jia, Sheth, and Svensson~\cite{DBLP:conf/ipco/JiaSS20}, using different techniques, gave a $3$-approximation algorithm for $\gckc$ running in time $|X|^{O(\gamma^2)}$.

In a second step we extend and generalize our technique to Fair~\gckc, which, as mentioned, is a generalization of \gckc. We show that Fair \gckc admits an $O(1)$-approximation, which neither violates covering nor probabilistic constraints.
\begin{theorem}\label{thm:mainFairGckc}
There is a $4$-approximation algorithm for Fair~\gckc running in time $\poly(L) \cdot |X|^{O(\gamma)}$, where $L$ is the encoding length of the input.
\end{theorem}

We complete our results by showing inapproximability of \gckc when $\gamma$ is not bounded. This holds even on the real line ($1$-dimensional Euclidean space).

\begin{theorem}\label{thm:mainHardness}
It is $\NP$-hard to decide whether \gckc on the real line admits a solution of radius $0$. Moreover, unless the Exponential Time Hypothesis fails, for any function $f:\mathbb{Z}_{> 0}\to \mathbb{Z}_{\geq 0}$ with $f(n) = \omega(\log n)$, no polynomial-time algorithm can distinguish whether \gckc on the real line with $\gamma \leq f(|X|)$ admits a solution of radius $0$.
\end{theorem}
Hence, assuming the Exponential Time Hypothesis, there is no polynomial-time approximation algorithm for \gckc if the number of colors grows faster than logarithmic in the size of the ground set. Notice that, for a logarithmic number of colors, our procedures run in quasi-polynomial time.

Finally, we extend the hardness implied by Theorem~\ref{thm:mainHardness} to bi-criteria algorithms. An $(\alpha,\beta)$ bi-criteria algorithm for \gckc, for $\alpha, \beta \geq 1$, is an algorithm that returns a solution that picks at most $\alpha k$ centers, and its radius is at most $\beta r$, where $r$ is the radius of an optimal solution with $k$ centers. More precisely, we prove the following theorem.
\begin{theorem}\label{thm:bicriteria-hardness}
There exists a constant $c > 0$,
such that it is $\NP$-hard to decide whether \gckc on the real line admits a solution of radius $0$, even if we are allowed to violate the number of centers to be opened by a factor of $c \log |X|$.
\end{theorem}

\subsection{Outline of main technical contributions and paper organization}

We introduce two main technical ingredients. The first is a method to deal with additional covering constraints in $k$-Center problems. We showcase this method in the context of \gckc, which leads to Theorem~\ref{thm:mainGckc}. For this, we combine polyhedral sparsity-based arguments as used by Bandyapadhyay et al.~\cite{DBLP:conf/esa/Bandyapadhyay0P19}, which by themselves only lead to pseudo-approximations, with dynamic programming to design a round-or-cut approach. Round-or-cut approaches, first used by Carr et al.~\cite{carr_2000_strengthening}, leverage the ellipsoid method in a clever way. In each ellipsoid iteration they either separate the current point from a well-defined polyhedron $P$, or round the current point to a good solution. The rounding step may happen even if the current point is not in $P$. Round-or-cut methods have found applications in numerous problem settings (see, e.g.,~\cite{levi_2008_approximation,li_2017_uniform,li_2016_approximating,an_2017_lp-based,nutov_2017_tree,grandoni_2018_improved,DBLP:journals/talg/ChakrabartyN19}).
The way we employ round-or-cut is inspired by a powerful round-or-cut approach of Chakrabarty and Negahbani~\cite{DBLP:journals/talg/ChakrabartyN19} also developed in the context of $k$-Center. However, their approach is not applicable to $k$-Center problems as soon as multiple covering constraints exist, like in \gckc; see Appendix~\ref{appendix:bad-example-deeparnab} for more details.

Our second technical contribution first employs LP duality to transform lottery-type models, like Fair~\gckc, into an auxiliary problem that corresponds to a weighted version of $k$-Center with covering constraints. We then show how a certain type of approximate separation over the dual is possible, by leveraging the techniques we introduced in the context of $\gckc$, leading to a $4$-approximation.

Even though Theorem~\ref{thm:mainFairGckc} is a strictly stronger statement than Theorem~\ref{thm:mainGckc}, we first prove Theorem~\ref{thm:mainGckc} in Section~\ref{sec:main_algorithm}, because it allows us to give a significantly cleaner presentation of some of our main technical contributions. In Section~\ref{sec:lottery}, we then focus on the additional techniques needed to deal with Fair~\gckc, by reducing it to a problem that can be tackled with the techniques introduced in Section~\ref{sec:main_algorithm}. Finally, in Section~\ref{sec:hardness}, we discuss the hardness results stated in Theorems~\ref{thm:mainHardness} and~\ref{thm:bicriteria-hardness}.

\section{A 4-approximation for \gckcB with running time $\bm{|X|^{O(\gamma)}}$}
\label{sec:main_algorithm}

In this section, we prove Theorem~\ref{thm:mainGckc}, which implies a polynomial-time $4$-approximation algorithm for \gckc with constantly many colors. We assume $\g \geq 2$; notice that $\g =1$ corresponds to Robust $k$-Center, for which an (optimal) polynomial-time $2$-approximation is known~\cite{DBLP:conf/icalp/ChakrabartyGK16,HarrisPST19}. Moreover, we assume that $\gamma < k$, since otherwise, we can simply enumerate over all subsets of $X$ of size $k$, which leads to an exact algorithm with running time $|X|^{O(k)} \leq |X|^{O(\gamma)}$. Thus, from now on, we have that $2 \leq \gamma \leq k - 1$.

We present a procedure that for any $r\in \mathbb{R}_{\geq 0}$ returns a solution of radius $4r$ if a solution of radius $r$ exists, and runs in time $|X|^{O(\gamma)}$. This implies Theorem~\ref{thm:mainGckc} because the optimal radius is a distance between two points. Hence, we can run the procedure for all possible pairwise distances $r$ between points in $X$ (or, alternatively, do binary search on the set of pairwise distances in order to speed up the algorithm) and return the best solution found. Thus, we fix $r\in \mathbb{R}_{\geq 0}$ in what follows.
We denote by $\p$ the following canonical relaxation of \gckc with radius $r$:
\begin{equation}\label{eq:defP}
\p = \left\{ (x,y)\in [0,1]^X \times [0,1]^X \: \middle| \:
\renewcommand\arraystretch{1.2}
\begin{array}{>{\displaystyle}rcl@{\quad}l}
    \sum_{v\in X}y(v)                &\leq &k & \\
    \sum_{v\in B(u,r)} y(v)  &\geq& x(u) & \forall u\in X \\ 
    \sum_{u\in X_{\ell}} x(u)             &\geq &m_{\ell} &\forall \ell \in [\g]  
\end{array}
\right\}\enspace .
\end{equation}

Integral points $(x,y)\in \p$ correspond to solutions of radius $r$, where $x$ and $y$ are characteristic vectors indicating the points that are covered and the centers that are opened, respectively. We denote the integer hull of $\p$ by $\ip\coloneqq \conv\left(\p \cap (\{0,1\}^X \times \{0,1\}^X )\right)$ .

Our algorithm is based on the round-or-cut framework, first used in~\cite{carr_2000_strengthening}. The main building block is a procedure that rounds a point $(x,y)\in \p$ to a radius $4r$ solution under certain conditions. It will turn out that these conditions are always satisfied if $(x,y) \in \ip$. If they are not satisfied, then we can prove that $(x,y) \notin \ip$ and generate in time $|X|^{O(\gamma)}$ a hyperplane separating $(x,y)$ from $\ip$. This separation step now becomes an iteration of the ellipsoid method, employed to find a point in $\ip$, and we continue with a new candidate point $(x,y)$. Schematically, the whole process is described in Figure~\ref{fig:flowchart}.

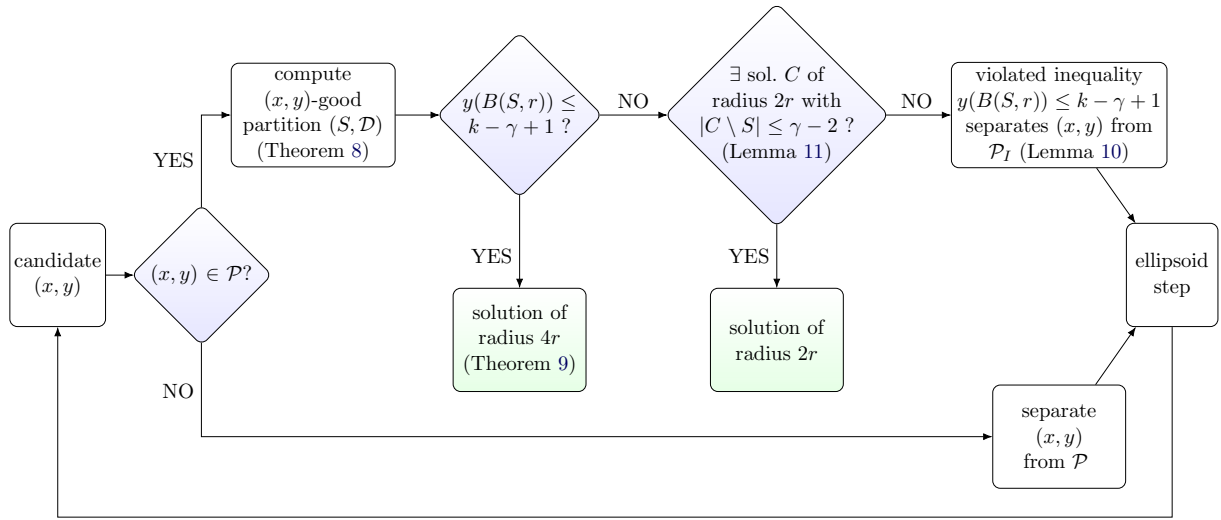
\begin{figure}[ht]
\begin{center}
    \hspace*{0.1em}\resizebox{.99\textwidth}{!}{%
\begin{tikzpicture}
  \matrix (chart)
    [
      matrix of nodes,
      nodes in empty cells,
      column sep      = 3em,
      row sep         = 0.4em,
      minimum height = 5em,
      ampersand replacement=\&
    ]
    {       \&[-1.8em]      \&[-5em] \node [treenode,text width=7.5em] (partition) {compute $(x,y)$-good partition $(S,\mathcal{D})$ (Theorem~\ref{thm:partition-algorithm})};  \&[-1em] \node [decision] (dec1) { $ y(B(S,r))\leq k-\gamma+1 $ ?}; \&  \node [decision,text width=8.0em, inner sep = -5pt] (dec2) {$\exists$ sol.~$C$ of \\ radius $2r$ with\\ $|C\setminus S| \leq \gamma - 2$ ?\\ (Lemma~\ref{lemma:DP})};    \&     \node [treenode,text width=10em] (sep2) {violated inequality $y(B(S,r))\leq k-\gamma+1$ separates $(x,y)$ from $\ip$ (Lemma~\ref{lemma:separate}) };    \&[-5em] \\[-1.5em]
    \node [treenode, text width=4.2em] (candidate) {candidate $(x,y)$};	\& \node [decision] (membership) {$(x,y)\in \mathcal{P}$?} ;  \& 	\&   \&     	\&\&	\node[treenode, text width=4em] (ellipsoid) {ellipsoid\\ step}; \\ [-3.2em]
    \&\&\&  	\node[env] (4app) {solution of radius $4r$ (Theorem~\ref{lemma:small-y-sum})};    \&     \node [env] (2app) {solution of radius $2r$};    \&   \\[-0.75em]
    	 \&\&\&\& \& \node [treenode] (sep1) {separate $(x,y)$ from $\mathcal{P}\ $ }; \& \\
    };

   \tikzstyle{every path}=[draw, -Latex]
   
   \begin{scope}[shorten >= -1.7pt]
    \path (candidate) -- (membership);
    \path (partition) -- (dec1);
   \end{scope}
   
   \begin{scope}[shorten <= -1.7pt]
   \path (membership) |-  node [left, near start] {NO}(sep1);
   \path (membership) |- node [left, near start] {YES} (partition);
   \path[shorten >= -1.7pt] (dec1) -- node [above]{NO} (dec2);
   \path (dec2)--node [above, xshift = -1.7pt] {NO} (sep2);
   \path (dec1)-- node[left, pos=1, yshift = 17pt] {YES}(4app);
   \path (dec2) -- node[left, pos=1, yshift = 17pt]{YES} (2app);
   \end{scope}
   \path (sep1) -- (ellipsoid);
   \path (sep2) -- (ellipsoid);
   \coordinate  (intersection) at (sep1 -| ellipsoid) {} ;
   \coordinate[below = 1.5cm of intersection] (bottomright) {};
   \draw (ellipsoid.south) -- (bottomright) -| (candidate);

\end{tikzpicture} %
}
\end{center}
\caption{An iteration of the ellipsoid method.}
\label{fig:flowchart}
\end{figure}

On a high level, we realize our round-or-cut procedure as follows. First, we check whether $(x,y) \in \p$ and return a violated constraint if this is not the case. If $(x,y)\in \p$, we partition the metric space, based on a natural greedy heuristic introduced by Harris et al.~\cite{HarrisPST19}. This gives a set of centers $S=\{s_1,\ldots , s_q\}\subseteq X$ with corresponding clusters $\mathcal{D}=\{D_1, \ldots, D_q\}\subseteq 2^X$. We now exploit a technique by Bandyapadhyay et al.~\cite{DBLP:conf/esa/Bandyapadhyay0P19}, which implies that if $y(B(S, r)) \leq k - \g + 1$, then one can leverage sparsity arguments in a simplified LP to obtain a radius $4r$ solution that picks centers only within $S$. We then turn to the case where $y(B(S, r)) > k - \g + 1$. 
At this point, we show that one can efficiently check whether there exists a solution of radius $2r$ that opens at most $k - (k - \gamma + 2) = \gamma - 2$ centers outside of $S$. This is achieved by guessing the centers outside of $S$ (of which there are at most $\g-2$ many, as noted) and using dynamic programming to find the remaining centers in $S$.
If no such radius $2r$ solution exists, we argue that any solution of radius $r$ has at most $k-\g+1$ centers in $B(S,r)$, proving that $y(B(S, r)) \leq k - \g + 1$ is an inequality separating $(x,y)$ from $\ip$.

\medskip

We now give a formal treatment of each step of this algorithm, which is schematically described in Figure~\ref{fig:flowchart}. Given a point $(x,y)\in \mathbb{R}^X \times \mathbb{R}^X$, we first check whether $(x,y)\in \p$, and, if not, return a violated constraint of $\p$. Such a constraint separates $(x,y)$ from $\ip$ because $\ip \subseteq \p$. Hence, we may assume
that $(x, y) \in \p$. 

We now use a partitioning technique by Harris et al.~\cite{HarrisPST19} that, given $(x, y) \in \p$, allows for obtaining what we call an $(x,y)$-good partition $(S,\mathcal{D})$, defined as follows.

\begin{definition}[$(x,y)$-good partition]
\label{def:good-partition}
Let $(x,y) \in \p$. A tuple $(S,\mathcal{D})$, where the family $\mathcal{D} = \{D_1, \ldots, D_q\}$ partitions $X$ and $S = \{s_1, \ldots, s_q\}\subseteq X$ with $s_i\in D_i$ for $i\in [q]$, is an $(x,y)$-good partition if:
\begin{enumerate}
    \item $d(s_i, s_j) > 4r$ for all $i,j\in [q], i\neq j$,
    \item $D_i \subseteq B(s_i, 4r)$ for all  $i \in [q]$, and
    \item \label{item:goodClusterColors} 
    $y(B(s_i,r)) \geq x(u)$ for all $i\in [q]$ and for all $u\in D_i$.    
\end{enumerate}
\end{definition}

The partitioning procedure of~\cite{HarrisPST19} was originally introduced for Robust $k$-Center and naturally extends to \gckc (see~\cite{DBLP:conf/esa/Bandyapadhyay0P19}). For completeness, we describe it in Algorithm~\ref{alg:partition}.
Contrary to prior procedures, we compute an $(x,y)$-good partition whose centers have pairwise distances of strictly more than $4r$ (instead of $2r$ as in prior work). This large separation avoids overlap of radius $2r$ balls around centers in $S$, and allows us to use dynamic programming (DP) to build a radius $2r$ solution with centers in $S$ under certain conditions. However, it is also the reason why we get a $4$-approximation if the DP approach cannot be applied.

\medskip
\begin{algorithm2e}[ht]
\SetAlgoLined
 $U \leftarrow X$; $\; i\leftarrow 0$\; 
 $S \leftarrow \emptyset$; $\; \mathcal{D} \leftarrow \emptyset$\;
 \While{$U \neq \emptyset$} {
 $i\leftarrow i+1$; $\;s_i \leftarrow \argmax_{u \in U} \{x(u)\}$\;
 $D_i \leftarrow B(s_i, 4r) \cap U$\;
 $S\leftarrow S \cup \{s_i\}$; $\; \mathcal{D} \leftarrow \mathcal{D} \cup \{D_i\}$\;
 $U\leftarrow U\setminus B(s_i,4r)$; 
 }
 \KwRet{$(S,\mathcal{D})$}
 \caption{Compute $(x,y)$-good partition, given $(x,y) \in \p$}
\label{alg:partition}
\end{algorithm2e}

\begin{lemma}[\cite{HarrisPST19,DBLP:conf/esa/Bandyapadhyay0P19}] \label{thm:partition-algorithm}
For $(x,y) \in \p$, Algorithm~\ref{alg:partition} computes an $(x,y)$-good partition $(S, \mathcal{D})$ in polynomial time.
\end{lemma}
For completeness, we present the proof of the above lemma.
\begin{proof}[Proof of Lemma~\ref{thm:partition-algorithm}]
By construction, the first two properties of the definition of an $(x,y)$-good partition are trivially satisfied by the generated partition $(S, \mathcal{D})$. We now turn to the third property. For each point $u \in D_i$, by the greedy criterion we have $x(u) \leq x(s_i)$. Since $(x,y) \in \p$, we also have $x(s_i) \leq y(B(s_i, r))$, implying the statement. 
\end{proof}
The following theorem follows from the results in~\cite{DBLP:conf/esa/Bandyapadhyay0P19}.
\begin{theorem}[\cite{DBLP:conf/esa/Bandyapadhyay0P19}]
\label{lemma:small-y-sum}
Let $(x, y) \in \p$ and $(S,\mathcal{D})$ be an $(x,y)$-good partition. Then, if $y(B(S, r)) \leq k - \gamma + 1$, a solution of radius $4r$ can be found in polynomial time.
\end{theorem}
For completeness, we provide in Appendix~\ref{appendix:missing-proofs} a proof of a slightly stronger version of Theorem~\ref{lemma:small-y-sum}, namely Theorem~\ref{thm:general-sparse}, which we reuse later in a more general context. Theorem~\ref{thm:general-sparse} easily follows by the same sparsity argument used in~\cite{DBLP:conf/esa/Bandyapadhyay0P19}.

We are left with the case $y(B(S, r)) > k - \gamma + 1$. If $(x,y) \in \ip$, then there must exist a solution $C_1\subseteq X$ of radius $r$ with $|C_1\cap B(S, r)| > k - \gamma + 1$. In particular, we must have $|C_1 \setminus B(S,r)| \leq \gamma - 2$. We observe that if such a solution $C_1$ exists, then there must be a solution $C_2$ of radius $2r$ which has at most $\gamma - 2$ centers outside of $S$. This is formalized in the following lemma.

\begin{lemma}\label{lemma:separate}
Let $S\subseteq X$ with $d(s,s')>4r$ for all $s,s'\in S$ with $s \neq s'$, and $\tau \in \{0, \ldots, k-1\}$. 
If there is a radius $r$ solution $C_1$ with $|C_1\cap B(S,r)| > \tau$, then there is a radius $2r$ solution $C_2$ with $|C_2\setminus S| \leq k - \tau - 1$.
\end{lemma}
\begin{proof}
Assume there is a solution $C_1$ of radius $r$ with $|C_1\cap B(S,r)| > \tau$. Let $A = C_1 \cap B(S,r)$. For each $p \in A$, let $\phi(p) \in S$ be the unique point in $S$ such that $p \in B(\phi(p), r)$; $\phi(p)$ is well defined because $d(s, s') > 4r$ for every $s \neq s' \in S$. Thus, $|\phi(A)| \leq |A|$, where $\phi(A) \coloneqq \{\phi(p):\,p \in A\}$.

Let $C_2 = \phi(A) \cup (C_1 \setminus A)$. We have $|C_2| = |\phi(A)| + |C_1 \setminus A| \leq |A| + |C_1 \setminus A| \leq k$. Moreover, as $d(p, \phi(p)) \leq r$ for every $p \in A$, we have that $B(C_1, r) \subseteq B(C_2, 2r)$. Thus, $C_2$ is a feasible solution of radius $2r$. Finally, by construction, $|C_2 \setminus S| = |C_1 \setminus B(S,r)| \leq k - \tau - 1$. 
\end{proof}

So, we have now proved that if $y(B(S,r)) > k-\gamma +1$ and $(x,y)\in \ip$, then there is a solution $C_2$ of radius $2r$ with $|C_2\setminus S|\leq \gamma - 2$.
The motivation for considering solutions of radius $2r$ with all centers in $S$ except for constantly many (if $\gamma=O(1)$) is that such solutions can be found efficiently via dynamic programming. This is possible because the centers in $S$ are separated by distances strictly larger than $4r$, which implies that radius $2r$ balls centered at points in $S$ do not overlap. Hence, there are no interactions between such balls. This is formalized below.
\begin{lemma}\label{lemma:DP}
Let $S\subseteq X$ with $d(s,s')>4r$ for all $s,s'\in S$ with $s \neq s'$, and $\beta \in \mathbb{Z}_{\geq 0}$.
If a radius $2r$ solution $C\subseteq X$ with $|C\setminus S|\leq \beta $ exists, then we can find such a solution in time $|X|^{O(\beta + \g)}$.
\end{lemma}
\begin{proof}
Suppose there is a solution $C \subseteq X$ of radius $2r$ with $|C \setminus S| \leq \beta $. The algorithm has two components. We first guess the set $Q\coloneqq C\setminus S$. Because $|Q| \leq \beta $, there are $|X|^{O(\beta )}$ choices. Given $Q$, it remains to select at most $k-|Q|$ centers $W\subseteq S$ to fulfill the color requirements. Note that for any $W\subseteq S$, the number of points of color $\ell\in [\gamma]$ that $B(W,2r)$ covers on top of those already covered by $B(Q,2r)$ is
$\left\vert(B(W,2r)\setminus B(Q,2r))\cap X_\ell\right\vert =
\sum_{w\in W} \left\vert \left(B(w,2r) \setminus B(Q,2r) \right) \cap X_\ell\right\vert,
$
where equality holds because centers in $W$ are separated by distances strictly larger than $4r$, and thus $B(W,2r)$ is the disjoint union of the sets $B(w,2r)$ for $w\in W$. Hence, the task of finding a set $W\subseteq S$ with $|W|\leq k-|Q|$ such that $Q\cup W$ is a solution of radius $2r$ can be phrased as finding a feasible solution to the following binary program:
\begin{equation}
\label{binary-program}
\begingroup
\renewcommand*{\arraystretch}{1.0}
\begin{array}{>{\displaystyle}rcl@{\quad}l}
\sum_{s\in S} z(s) \cdot \left\vert (B(s,2r) \setminus B(Q,2r)) \cap X_\ell \right\vert &\geq & m_\ell - |B(Q,2r)\cap X_\ell|  &\forall \ell\in [\gamma]\\
\sum_{s\in S} z(s)  &\leq  &k-|Q| &\\
z &\in &\{0,1\}^S \enspace . &
\end{array}
\endgroup
\end{equation}
The above binary program can be easily solved through standard dynamic programming techniques in $|X|^{O(\gamma)}$ time, because the coefficients are small. For completeness, we show in Appendix~\ref{appendix:missing-proofs} how this can be done for a slightly more general problem (see Theorem~\ref{thm:dp}), which we will reuse later on.\footnote{Program~(\ref{binary-program}) reduces to the one of Theorem~\ref{thm:dp}  by removing any redundant constraint of the first type that has negative right-hand side.} As the dynamic program is run for $|X|^{O(\beta)}$ many guesses of $Q$, we obtain an overall running time of $|X|^{O(\beta+\gamma)}$, as claimed. 
\end{proof}

This completes the last ingredient for an iteration of our round-or-cut approach as shown in Figure~\ref{fig:flowchart}. In summary, assuming $y(B(S,r)) > k-\gamma +1$ (for otherwise Theorem~\ref{lemma:small-y-sum} leads to a solution of radius $4r$) we use Lemma~\ref{lemma:DP} to check whether there is a radius $2r$ solution $C_2$ with $|C_2\setminus S|\leq \gamma-2$. This requires $|X|^{O(\gamma)}$ time. If this is the case, we are done. If not, the contrapositive of Lemma~\ref{lemma:separate} (with $\tau =k - \g + 1$) implies that every radius $r$ solution $C_1$ fulfills $|C_1\cap B(S,r)| \leq k - \gamma + 1$. Hence, every point $(\overline{x},\overline{y})\in \ip$ satisfies $\overline{y}(B(S,r)) \leq k - \gamma +1$. However, this constraint is violated by $(x,y)$, and so it separates $(x,y)$ from $\ip$. Thus, we proved that the process described in Figure~\ref{fig:flowchart} is a valid round-or-cut procedure that runs in time $|X|^{O(\gamma)}$.
\begin{corollary}
There is an algorithm that, given a point $(x,y)\in \mathbb{R}^X \times \mathbb{R}^X$, either returns a \gckc solution of radius $4r$ or an inequality separating $(x,y)$ from $\ip$. The running time of the algorithm is $|X|^{O(\gamma)}$.
\end{corollary}

We can now prove the main theorem.
\begin{proof}[Proof of Theorem~\ref{thm:mainGckc}]
We run the ellipsoid method on $\ip$ for each of the $O(|X|^2)$ candidate radii $r$.
For each $r$, the number of ellipsoid iterations is polynomially bounded as the separating hyperplanes that are produced by the algorithm have encoding length at most $O(|X|)$ (see Theorem 6.4.9 of~\cite{schrijver2012geometric}). To see this, note that all generated hyperplanes are either inequalities defining $\p$ or inequalities of the form $y(B(S,r))\leq k-\g+1$. For the correct guess of $r$, $\ip$ is non-empty and the algorithm terminates by returning a radius $4r$ solution. Hence, if we return the best solution among those computed for all guesses of $r$, we have a $4$-approximation, and the total running time is $\poly(|X|) \cdot |X|^{O(\gamma)} = |X|^{O(\gamma)}$.  
\end{proof}

\section{The lottery model of Harris et al.~\cite{HarrisPST19}}\label{sec:lottery}

Our main tool to solve the lottery model of Harris et al.~\cite{HarrisPST19} is a reduction to a certain type of weighted $k$-center problem. A key step of this reduction is to transform the problem through the use of linear duality. In Subsection~\ref{subsec:lottRedWeighted}, we first present this reduction before proving in Subsection~\ref{app:lottery} our algorithmic result for the above-referred version of a weighted $k$-center problem.

\subsection{Reduction to weighted version of $
\gckcB$}\label{subsec:lottRedWeighted}

Let $(X,d)$ be a Fair \gckc instance, and let $\mathcal{F}(r)$ be the family of sets of centers satisfying the covering requirements with radius $r$, i.e.,
\begin{equation*}
\mathcal{F}(r)\coloneqq \big\{C \subseteq X \,\big\vert\, |C| \leq k \textrm{ and } |B(C,r) \cap X_\ell| \geq m_\ell \;\;\forall \ell\in [\gamma]\big\} \enspace.
\end{equation*}
Note that a radius $r$ solution for Fair~$\gckc$ defines a distribution over the sets in $\mathcal{F}(r)$. Given $r$, such a distribution exists if and only if the following (exponential-size) linear program $\PLP(r)$ is feasible (with $\DLP(r)$ being its dual):

\renewcommand\arraystretch{1.3}
{ \everymath={\displaystyle}
\begin{equation*}
\begin{array}{lrcl>{\quad}ll}

       \PLP(r):&\min &  0 \hspace*{40pt}& & \\
    
  &\mathrm{s.t.}&\sum_{\substack{C \in \mathcal{F}(r):\\ u \in B(C,r)}}\lambda(C) &\geq &p(u)   &\forall u \in X\\
 	 &&\sum_{C\in \mathcal{F}(r)} \lambda(C)  &=&1&\\
 	 &&\lambda &\in &\mathbb{R}_{\geq 0}^{\mathcal{F}(r)}&\\
 	 &&&&&\\
 	    \DLP(r):&\max & \sum_{u \in X} p(u)\alpha(u) - \mu &&\\
    &\mathrm{s.t.}& \sum_{u \in B(C,r)} \alpha(u) &\leq &\mu &\forall C \in \mathcal{F}(r)\\
    && \alpha &\in &\mathbb{R}_{\geq 0}^{X}\\
    && \mu &\in &\mathbb{R}\enspace.
    
\end{array}
\end{equation*}

}
\renewcommand\arraystretch{1.0}

Clearly, if $\PLP(r)$ is feasible, then its optimal value is $0$. As mentioned in the introduction, it is also easy to see that if  $\PLP(r)$ is feasible, then it has a feasible solution with polynomial support (since the number of non-trivial constraints is $|X| + 1$).

We will again assume that $\gamma < k$. If $\gamma \geq k$, then for each fixed radius $r$, we solve $\PLP(r)$ in time $\poly(L)\cdot |X|^{O(k)} \leq \poly(L)\cdot |X|^{O(\gamma)}$, where $L$ is the encoding length of the input. If $\PLP(r)$ is infeasible, then the radius $r$ is too small. Otherwise, we compute a feasible extreme point solution to $\PLP(r)$ which corresponds to a distribution with support size $\poly(|X|)$. Hence, by applying binary search over all candidate radii, which are the $O(|X|^2)$ pairwise distances between points in $X$, we can compute an optimal distribution for the smallest possible radius in $\poly(L) \cdot |X|^{O(\gamma)}$ time. Thus, from now on, we assume that $1 \leq \gamma < k$.

Note that, for any $r \geq 0$, $\DLP(r)$ always has a feasible solution (the zero vector) of value $0$. Thus, by strong duality, $\PLP(r)$ is feasible if and only if the optimal value of $\DLP(r)$ is $0$. We note now that $\DLP(r)$ is scale-invariant, meaning that if $(\alpha, \mu)$ is feasible for $\DLP(r)$ then so is $(t\alpha, t\mu)$ for $t\in \mathbb{R}_{\geq 0}$. Hence, $\DLP(r)$ has a solution of strictly positive objective value if and only if $\DLP(r)$ is unbounded. We thus define the following polyhedron $\q(r)$, which contains all solutions of $\DLP(r)$ of value at least $1$:
\begingroup
\everymath={\displaystyle}
\begin{equation*}\label{eq:Qr}
   \q(r) \coloneqq \left\{
   (\alpha, \mu) \in \mathbb{R}_{\geq 0}^{X} \times \mathbb{R} \, \middle \vert \, 
   \begin{array}{rcl@{\quad}l}
   \sum_{u \in X} p(u)\alpha(u) &\geq& \mu + 1\\
    \sum_{u \in B(C,r)}^{} \alpha(u) &\leq& \mu &\forall C \in \mathcal{F}(r)
    \end{array}
    \right\}\enspace.
\end{equation*}
\endgroup
As discussed, the following statement is a direct consequence of strong duality of linear programming.
\begin{lemma}\label{lemma:empty-or-zero}
$\q(r)$ is empty if and only if $PLP(r)$ is feasible.%
\end{lemma}

The main lemma that allows us to obtain our result is the following. It guarantees the existence of an algorithm approximately solving a certain weighted $k$-center problem, where clients are weighted by $\alpha\in \mathbb{Q}^X_{\geq 0}$. Before proving the lemma in Subsection~\ref{app:lottery}, we show that it implies Theorem~\ref{thm:mainFairGckc}.

\begin{lemma}
\label{lemma:separate-or-nonempty_lottery}
    There is an algorithm that, given a point $(\alpha,\mu) \in \mathbb{Q}_{\geq 0}^{X} \times \mathbb{Q}$ satisfying $\sum_{u\in X}p(u)\alpha(u) \geq \mu + 1$ and a radius $r \geq 0$, either certifies that $(\alpha, \mu) \in \q(r)$, or outputs a set $C\in \mathcal{F}(4r)$ with $\sum_{u\in B(C,4r)} \alpha(u) > \mu$. The running time of the algorithm is $\poly(L) \cdot |X|^{O(\gamma)}$, where $L$ is the encoding length of the input.
\end{lemma}
In words, Lemma~\ref{lemma:separate-or-nonempty_lottery} either certifies $(\alpha,\mu)\in \q(r)$ or returns a hyperplane separating $(\alpha,\mu)$ from $\q(4r)$. Its proof leverages techniques introduced in Section~\ref{sec:main_algorithm}, and we present it in Subsection~\ref{app:lottery}.
Using Lemma~\ref{lemma:separate-or-nonempty_lottery}, we can now prove Theorem~\ref{thm:mainFairGckc}.
\begin{proof}[Proof of Theorem~\ref{thm:mainFairGckc}]
As noted, there are polynomially many choices for the radius $r$, for each of which we run the ellipsoid method to check emptiness of $\q(4r)$ as follows. Whenever there is a call to the separation oracle for a point $(\alpha,\mu)\in \mathbb{Q}^X \times \mathbb{Q}$, we first check whether $\alpha \geq 0$ and $\sum_{u\in X} p(u)\alpha(u) \geq \mu+1$. If one of these constraints is violated, we return it as separating hyperplane. Otherwise, we invoke the algorithm of Lemma~\ref{lemma:separate-or-nonempty_lottery}.
The algorithm either returns a constraint in the inequality description of $\q(4r)$ violated by $(\alpha,\mu)$, which solves the separation problem, or certifies $(\alpha,\mu)\in \q(r)$. If, at any iteration of the ellipsoid method, the separation oracle is called for a point $(\alpha,\mu)$ for which Lemma~\ref{lemma:separate-or-nonempty_lottery} certifies $(\alpha,\mu) \in \q(r)$, then Lemma~\ref{lemma:empty-or-zero} implies $\PLP(r)$ is infeasible. Thus, there is no solution to the considered Fair \gckc instance of radius $r$.
Hence, consider from now on that the separation oracle always returns a separating hyperplane, in which case the ellipsoid method certifies that $\q(4r) = \emptyset$ as follows.
Let $\mathcal{H}\subseteq \mathcal{F}(4r)$ be the family of all sets $C\in \mathcal{F}(4r)$ returned by Lemma~\ref{lemma:separate-or-nonempty_lottery} through calls to the separation oracle. Then, the following polyhedron:
\begingroup
\everymath={\displaystyle}
\begin{equation*}
\mathcal{Q}_{\mathcal{H}}(4r) = 
\left\{
(\alpha, \mu) \in \mathbb{R}_{\geq 0}^{X} \times \mathbb{R}
\, \middle \vert \,  
\begin{array}{rcl@{\quad}l}
\sum_{u \in X} p(u)\alpha(u) & \geq & \mu + 1\\
\sum_{u \in B(C,4r)}^{} \alpha(u) &\leq& \mu &\forall C \in \mathcal{H}
\end{array}
\right\}
\enspace,
\end{equation*}
\endgroup
which clearly contains $\mathcal{Q}(4r)$, is empty. As the encoding length of any constraint in the inequality description of $\mathcal{Q}(4r)$ is polynomially bounded in the input, the ellipsoid method runs in polynomial time (see Theorem 6.4.9 of~\cite{schrijver2012geometric}). In particular, the number of calls to the separation oracle, and thus $|\mathcal{H}|$, is polynomially bounded.

As $\mathcal{Q}(4r) \subseteq \mathcal{Q}_{\mathcal{H}}(4r) = \emptyset$, Lemma~\ref{lemma:empty-or-zero} implies that $PLP(4r)$ is feasible. More precisely, because $Q_{\mathcal{H}}(4r)=\emptyset$, the linear program obtained from $DLP(4r)$ by replacing $\mathcal{F}(4r)$, which parameterizes the constraints in $DLP(4r)$, by $\mathcal{H}$, has optimal value equal to $0$. Hence, its dual, which corresponds to $PLP(4r)$ where we replace $\mathcal{F}(4r)$ by $\mathcal{H}$, is feasible. As this feasible linear program has polynomial size, because $|\mathcal{H}|$ is polynomially bounded, we can solve it efficiently to obtain a distribution with the desired properties. Moreover, the total running time is $\poly(L) \cdot |X|^{O(\gamma)}$, where $L$ is the encoding length of the input.  
\end{proof}

\subsection{Proof of Lemma~\ref{lemma:separate-or-nonempty_lottery}}\label{app:lottery}

The desired separation algorithm requires us to find a solution for a $\gckc$ instance with an extra covering constraint; the procedure of Section~\ref{sec:main_algorithm} generalizes to handle this extra constraint. We follow similar steps as in Figure~\ref{fig:flowchart}.

Let $(\alpha,\mu) \in \mathbb{Q}_{\geq 0}^{X} \times \mathbb{Q}$ be a point satisfying $\sum_{u\in X}p(u)a(u) \geq \mu + 1$, let $r \geq 0$, and, moreover, let
\begin{equation*}
\mathcal{F}^{\alpha,\mu}(r) \coloneqq \left\{C\in \mathcal{F}(r)\, \middle\vert\, \sum_{u\in B(C,r)} \alpha(u) >\mu \right\}\enspace.
\end{equation*}
Hence, to prove Lemma~\ref{lemma:separate-or-nonempty_lottery}, we need to find a procedure that either certifies $\mathcal{F}^{\alpha,\mu}(r)=\emptyset$ or returns a set $C\in \mathcal{F}^{\alpha,\mu}(4r)$.
To avoid technical complications later on due to the strict inequality in the definition of $\mathcal{F}^{\alpha,\mu}(r)$, we observe, using standard techniques, that one can efficiently compute a polynomially encoded $\varepsilon >0$ to replace the inequality $\sum_{u\in B(C,r)}\alpha(u) > \mu$ by $\sum_{u\in B(C,r)} \alpha(u) \geq \mu + \epsilon$.
\begin{lemma}\label{lemma:eps-gap}
Let $(\alpha,\mu) \in \mathbb{Q}_{\geq 0}^X \times \mathbb{Q}$. 
Then one can efficiently compute an $\varepsilon > 0$  with encoding length $O(L)$, where $L$ is the encoding length of $(\alpha,\mu)$, such that the following holds: For any $C\in \mathcal{F}(r)$, we have $\sum_{u\in B(C,r)}\alpha(u) >\mu$ if and only if $\sum_{u\in B(C,r)} \alpha(u) \geq \mu+\varepsilon$. 
\end{lemma}
\begin{proof}
The tuple $(\alpha,\mu)$ consists of $|X|+1$ rationals $\left\{\sfrac{p_i}{q_i} \right\}_{i \in [N]}$, with $p_i \in \mathbb{Z}$ and $q_i \in \mathbb{Z}_{>0}$. Let $\Pi = \prod_{i \in [N]} q_i$. Note that if $\sum_{u\in B(C,r)} \alpha(u) > \mu$, then $\sum_{u\in B(C,r)} \alpha(u) - \mu \geq \frac{1}{\Pi}$. Thus, we set $\varepsilon = \sfrac{1}{\Pi}$. Moreover $\log \Pi = \sum_{i \in [N]} \log q_i$, and so the encoding length of $\varepsilon$ is $O(L)$.   
\end{proof}

Let $\fp$ be the following modified relaxation of \gckc, defined for given $(\alpha, \mu) \in \mathbb{Q}_{\geq 0}^X \times \mathbb{Q}$, and a corresponding $\varepsilon > 0$ as per Lemma~\ref{lemma:eps-gap}, where the polytope $\p$ is defined for a fixed radius $r$, as in Section~\ref{sec:main_algorithm} (see~\eqref{eq:defP}):
\begin{equation*}
\fp \coloneqq \left\{ (x,y) \in \p \,\middle\vert\, \;\sum_{u\in X} \alpha(u) x(u)  \geq \mu +\varepsilon\right\}\enspace. 
\end{equation*}
Let $\fip \coloneqq \conv\left(\fp \cap (\{0,1\}^X \times \{0,1\}^X)\right)$ be the integer hull of $\fp$. We now state the following straightforward observation, whose proof is an immediate consequence of the definitions of the corresponding polytopes and Lemma~\ref{lemma:eps-gap}.
\begin{observation}\label{lemma:aux-equiv}
Let $(\alpha, \mu)\in \mathbb{Q}_{\geq 0}^X \times \mathbb{Q}$ be such that $\sum_{u\in X} p(u)\alpha(u) \geq \mu + 1$ and $\fip=\emptyset$. Then $(\alpha, \mu)\in \q(r)$.
\end{observation}

The following lemma is a slightly modified version of Theorem~\ref{lemma:small-y-sum}, which is also a direct consequence of Theorem~\ref{thm:general-sparse} given in Appendix~\ref{appendix:missing-proofs}.
\begin{lemma}
\label{lemma:small-y-sum_lottery}
Let $(\alpha,\mu) \in \mathbb{Q}_{\geq0}^X\times\mathbb{Q}$, let $(x, y) \in \fp$, and let $(S, \mathcal{D})$ be an $(x,y)$-good partition. If $y(B(S, r)) \leq k - \gamma$, a set $C\in \mathcal{F}^{\alpha,\mu}(4r)$ can be found in polynomial time.
\end{lemma}

If $y(B(S,r)) \leq k - \gamma$, then Lemma~\ref{lemma:small-y-sum_lottery} leads to a set $C \in \mathcal{F}(4r)$ that satisfies $\sum_{u \in B(C,4r)} \alpha(u) > \mu $; this gives a constraint separating $(\alpha,\mu)$ from $\q(4r)$.

It remains to consider the case $y(B(S,r))>k-\g$. As in Section~\ref{sec:main_algorithm}, we can either find a set $C_2\in \mathcal{F}^{\alpha, \mu}(2r)$ or certify that every $C_1\in \mathcal{F}^{\alpha,\mu}(r)$ satisfies $|C_1\cap B(S,r)|\leq k-\g$.

\begin{lemma}\label{lemma:separate_lottery}
Let $(\alpha,\mu) \in \mathbb{Q}_{\geq0}^X\times\mathbb{Q}$, $S\subseteq X$ with $d(s,s')>4r$ for all $s, s' \in S$ with $s\neq s'$, and $\tau \in \{0, \ldots, k-1\}$. If there is a set $C_1\in \mathcal{F}^{\alpha,\mu}(r)$ with $|C_1\cap B(S,r)|> \tau$, then there is a set $C_2\in \mathcal{F}^{\alpha,\mu}(2r)$ with $|C_2 \setminus S| \leq k - \tau - 1$.
\end{lemma}
%
%
%
%
%
%
The proof of the above lemma is identical to the proof of Lemma~\ref{lemma:separate}, and thus is omitted.

\begin{lemma}\label{lemma:DP_lottery}
Let $(\alpha,\mu) \in \mathbb{Q}_{\geq0}^X\times\mathbb{Q}$, $S\subseteq X$ with $d(s,s')>4r$ for all $s,s'\in S$ with $s \neq s'$, and $\beta \in \mathbb{Z}_{\geq 0}$.
If there exists a set $C\in  \mathcal{F}^{\alpha,\mu}(2r)$ with $|C\setminus S|\leq \beta $, then we can find such a set in time $|X|^{O(\beta + \g)}$.
\end{lemma}

\begin{proof}
As in the proof of Lemma~\ref{lemma:DP}, we first guess up to $\beta $ centers $Q \subseteq X\setminus S$. For each of those guesses, we consider the binary program~\eqref{binary-program} with objective function  
$\sum_{s\in S} z(s) \cdot \alpha(B(s,2r) \setminus B(Q,2r))$ to be maximized. Again, this is a special case of the binary program presented in Theorem~\ref{thm:dp}, given in Appendix~\ref{appendix:missing-proofs}, and thus can be solved in time $|X|^{O(\gamma)}$. For the guess $Q=C\setminus S$, the characteristic vector $\chi^{C\cap S}$ is feasible for this binary program, implying that the optimal centers $Z\subseteq S$ chosen by the binary program fulfill $Z\cup Q \in \mathcal{F}^{\alpha,\mu}(2r)$.   
\end{proof}

\begin{corollary}\label{cor:inner-ellipsoid-sep}
Let $(\alpha,\mu) \in \mathbb{Q}_{\geq0}^X\times\mathbb{Q}$. There is an algorithm that, given $(x,y)\in \mathbb{R}^X \times \mathbb{R}^X$, either returns a set $C\in \mathcal{F}^{\alpha,\mu}(4r)$ or returns a hyperplane separating $(x,y)$ from $\fip$. The running time of the algorithm is $\poly(L) \cdot |X|^{O(\gamma)}$, where $L$ is the encoding length of the input.
\end{corollary}

\begin{proof}
If $(x,y)\notin \fp$, we return a violated constraint separating $(x,y)$ from $\fp \supseteq \fip$. Hence we assume $(x,y)\in \fp$. Since $\fp \subseteq \p$, we can use Theorem~\ref{thm:partition-algorithm} to get an $(x,y)$-good partition $(S,\mathcal{D})$.
If $y(B(S,r))\leq k-\g$, Lemma~\ref{lemma:small-y-sum_lottery} gives a set $C\in \mathcal{F}^{\alpha,\mu}(4r)$. So, assuming $y(B(S,r))> k-\g$, we use Lemma~\ref{lemma:DP_lottery} (with $\beta=\gamma -1$) to check whether there is $C_2\in \mathcal{F}^{\alpha,\mu}(2r)$ with $|C_2\setminus S| \leq \g-1$. If this is the case, we are done because $\mathcal{F}^{\alpha,\mu}(2r)\subseteq \mathcal{F}^{\alpha,\mu}(4r)$. If not, the contrapositive of Lemma~\ref{lemma:separate_lottery} (with $\tau=\g-1$) implies that every $C_1\in \mathcal{F}^{\alpha,\mu}(r)$ fulfills $|C_1\cap B(S,r)|\leq k-\g$. Hence, every point $(\overline{x},\overline{y})\in \fip$ satisfies $\overline{y}(B(S,r))\leq k-\g$. However, this constraint is violated by $(x,y)$, and it thus separates $(x,y)$ from $\fip$.   
\end{proof}

\begin{proof}[Proof of Lemma~\ref{lemma:separate-or-nonempty_lottery}]
We use the ellipsoid method to check emptiness of $\fip$. Whenever the separation oracle gets called for a point $(x,y)\in \mathbb{R}^X \times \mathbb{R}^X$, we invoke the algorithm of Corollary~\ref{cor:inner-ellipsoid-sep}. If the algorithm returns at any point a set $C\in \mathcal{F}^{\alpha,\mu}(4r)$, then $C$ corresponds to a constraint in the inequality description of $\q(4r)$ violated by $(\alpha,\mu)$. Otherwise, the ellipsoid method certifies that $\fip=\emptyset$, which implies $(\alpha,\mu)\in \q(r)$ by Observation~\ref{lemma:aux-equiv}. Note that the number of iterations of the ellipsoid method is polynomial as the separating hyperplanes used by the procedure above have encoding length $\poly(L)$, where $L$ is the encoding length of the input (see Theorem 6.4.9 of~\cite{schrijver2012geometric}). Thus, the total running time is $\poly(L) \cdot |X|^{O(\gamma)}$, where $L$ is the encoding length of the input.   
\end{proof}

\section{Hardness results for Colorful $k$-Center}
\label{sec:hardness}

 We now prove our hardness results. We start in Subsection~\ref{subsec:mainHardness} by showing Theorem~\ref{thm:mainHardness}, i.e., that $\gckc$ becomes hard to approximate when the number of colors is unbounded. In Subsection~\ref{subsec:bicriteria}, we then prove   Theorem~\ref{thm:bicriteria-hardness}, which shows our bi-criteria inapproximability result, i.e., there is an approximation hardness even when one is allowed to slightly exceed the number of centers to be opened.

\subsection{Hardness of approximation for $\gckc$}\label{subsec:mainHardness}
In this section, we prove our main hardness result, Theorem~\ref{thm:mainHardness}. For that, we use  a reduction from Vertex Cover on graphs of maximum degree $3$.
\begin{definition}
Let $G = (V, E)$ be a graph of maximum degree $3$, and let $t \in [n]$, where $n = |V|$. The decision version of Vertex Cover asks to decide whether there exists a vertex cover $S \subseteq V$ of size at most $t$. We denote this problem as $\VCthree(G, t)$.
\end{definition}
We will use the following hardness results for $\VCthree(G, t)$.
\begin{theorem}[\cite{Garey:1974:SNP:800119.803884,DBLP:journals/jcss/CaiJ03}]\label{thm:VC-hardness}
$ $

\begin{enumerate}
    \item There is no algorithm for $\VCthree(G, t)$ that runs in polynomial time, assuming that $\P \neq \NP$.\label{thm: hardness-i}
    \item There is no algorithm for $\VCthree(G,t)$ that runs in time $2^{o(t)} \poly(n)$, assuming the Exponential Time Hypothesis.\label{thm: hardness-ii}
\end{enumerate}
\end{theorem}

Our reduction is now described in the next lemma.
\begin{lemma}\label{lemma:reduction}
Given a $\VCthree(G,t)$ instance, there exists a polynomial-time algorithm that constructs a \gckc instance on the real line with the following properties:
\begin{enumerate}
    \item $|X| = |V|$, $\g = |E|$, and $k = t$,
    \item it has a feasible solution of radius $0$ if and only if $G$ has a vertex cover of size at most $t$.
\end{enumerate}
\end{lemma}
\begin{proof}
Let $G= (V, E)$ be the graph of the given $\VCthree(G,t)$ instance. We construct a \gckc instance with $|X| = |V|$, $\g = |E|$, and $k = t$ as follows. Let $V = \{u_1,\ldots, u_{|V|}\}$ and $E = \{e_1, \ldots, e_{|E|}\}$. We set $X = \{1, \ldots, |V|\} \subseteq \mathbb{R}$. Each edge $e_\ell = \{u_i, u_j\} \in E$ corresponds to a distinct color $X_\ell = \{i, j\}$. Note that the number of colors of a point $i \in X$ is equal to the degree of the vertex $u_i$ in $G$. We also set the covering requirement for each color to be $m_\ell = 1$. A sample reduction is given in Figure~\ref{fig:hardness}. It is easy to observe that $\VC(G,t)$ is a YES instance if and only if there is a solution to our $\gckc$ instance of radius $0$. This implies the result.    
\end{proof}
\begin{figure}[ht]
\captionsetup{width=.8\linewidth}
\begin{center}
{%
		\begin{tikzpicture}
		\pgfmathsetmacro{\myheight}{2.5}
		\pgfmathsetmacro{\myshift}{0.8}

		\colorlet{c1}{orange!70!yellow}
		\colorlet{c2}{blue!70!black}
		\colorlet{c3}{green!50!black}
		\colorlet{c4}{red}

		\draw[->, line width = 0.5pt, black] (2,1) -- (7,1) {{}};
		
		\foreach \i in {1,...,4}
		{
			\draw[-,line width = 0.5pt, black] (2+\i,1.1) -- (2+\i,0.9)  {{}};
			\node[right] at (1.8+\i,0.6) {\small  \i};
		}

		\node[label={[label distance=0cm]270 :$v_1$}, circle,draw=black,minimum size=10] (v1) at (-1,1) {};
		\node[label={[label distance=0cm]180 :$v_2$}, circle,draw=black,minimum size=10] (v2) at (-2.5,2.5) {};
		\node[label={[label distance=0cm]90 :$v_3$}, circle,draw=black,minimum size=10] (v3) at (-1,4) {};
		\node[label={[label distance=0cm]0 :$v_4$}, circle,draw=black,minimum size=10] (v4) at (0.5,2.5) {};

\begin{scope}[line width=2.5pt]
		\draw[c2] (v2) -- node [midway,above left= 0pt, black] {$b$}  (v3) {{}};	
		\draw[c4] (v1) -- node [midway,below right= -0.9pt, black] {$d$} (v4) {{}};
		\draw[c1] (v1) -- node [midway,below left = +0.8pt, black] {$a$} (v2) {{}};	
		\draw[c3] (v1) -- node [text width=0.4cm,midway, black, left = -7pt] {$c$} (v3) {{}};	
\end{scope}

\begin{scope}[every node/.style={circle,draw=black,minimum size=14,inner sep=0pt}]		
		\node[fill=c1!30] (v1) at (4,1.5) {$a$};
		\node[fill=c1!30] (v1) at (3,1.5) {$a$};

		\node[fill=c2!30] (v1) at (4,2.1) {$b$};
		\node[fill=c2!30] (v1) at (5,1.5) {$b$};
		
		\node[fill=c3!30] (v1) at (3,2.1) {$c$};
		\node[fill=c3!30] (v1) at (5,2.1) {$c$};
		
		\node[fill=c4!30] (v1) at (3,2.7) {$d$};
		\node[fill=c4!30] (v1) at (6,1.5) {$d$};
\end{scope}

		\draw[-, line width = 2pt, black] (2.9,0.9) -- (3.1,1.1) {{}};
		\draw[-, line width = 2pt, black] (3.1,0.9) -- (2.9,1.1) {{}};		
				
		\draw[-, line width = 2pt, black] (4.9,0.9) -- (5.1,1.1) {{}};
		\draw[-, line width = 2pt, black] (5.1,0.9) -- (4.9,1.1) {{}};

		\draw[-, line width = 2pt, black] (-1.1,3.9) -- (-0.9,4.1) {{}};
		\draw[-, line width = 2pt, black] (-0.9,3.9) -- (-1.1,4.1) {{}};		
		\draw[-, line width = 2pt, black] (-1.1,0.9) -- (-0.9,1.1) {{}};
		\draw[-, line width = 2pt, black] (-0.9,0.9) -- (-1.1,1.1) {{}};

		\begin{scope}[xshift=-13mm,yshift=-2mm]
		\draw[-, line width = 2pt, black] (-2.1,-0.1) -- (-1.9,0.1) {{}};
		\draw[-, line width = 2pt, black] (-1.9,-0.1) -- (-2.1,0.1) {{}};
			
		\node[right] at (-1.8,-0.02) {: optimal solution for Vertex Cover and $\gckc$, respectively.};
		\end{scope}
		\end{tikzpicture}
}		
\caption{An example of the reduction for the input graph (with colored edges) on the left with the corresponding instance of $\gckc$ on the right.}
\label{fig:hardness}
\end{center}
\end{figure}%
We are now ready to prove Theorem~\ref{thm:mainHardness}.
\begin{proof}[Proof of Theorem~\ref{thm:mainHardness}]
The first part of the theorem is an immediate consequence of Lemma~\ref{lemma:reduction} and part~\ref{thm: hardness-i} of Theorem~\ref{thm:VC-hardness}.

For the second part, let $f: \mathbb{Z}_{\geq 0} \to \mathbb{Z}_{\geq 0}$ be a function that satisfies $f(n) = \omega(\log n)$ and assume for the sake of contradiction that there is a polynomial-time algorithm $\mathcal{A}$ that can distinguish whether \gckc on the real line with $\gamma \leq f(|X|)$ colors admits a solution of radius $0$. Let $c_1,c_2\in \mathbb{R}_{>0}$ be constants such that the running time of $\mathcal{A}$ is upper bounded by $c_1\cdot |X|^{c_2}$.
We will show how one can use $\mathcal{A}$ to solve $\VCthree(G,t)$ in time $2^{o(t)}\poly(|V|)$, where $G=(V,E)$ is an undirected graph and $t\in [|V|]$. For that, we use the reduction of Lemma~\ref{lemma:reduction}. We end up with an instance $\mathcal{I}$ of \gckc on a space $X\subseteq \mathbb{R}$ with $|X| = |V|$, $\gamma = |E|$ many colors, and $k = t$. Wlog, we can assume that $G=(V,E)$ is connected, which in particular implies $|E|\geq |X|-1$.

For algorithm $\mathcal{A}$ to be applicable we need to make sure that the number of colors $\gamma$ is in the correct regime with respect to the size of the ground set $|X|$. More precisely, we would need $\gamma \leq f(|X|)$. As this property may not hold for $\mathcal{I}$, we create an auxiliary \gckc instance $\overline{\mathcal{I}}$ obtained by inflating $\mathcal{I}$ through the addition of dummy vertices as discussed in the following.
Because $f(n) = \omega(\log n)$, there is a constant $n_0\in \mathbb{Z}_{>0}$ and a non-decreasing function $h:\mathbb{Z}_{\geq 0}\to \mathbb{Z}_{> 0}$ with
\begin{enumerate}[topsep=0.3em]
\item $\lim_{n\to\infty} h(n) = \infty$, and
\item $f(n) \geq h(n) \cdot \log n \quad\forall n\in \mathbb{Z}_{\geq n_0}$.
\end{enumerate}
Without loss of generality, we assume that $|X|\geq n_0$, as otherwise, the \gckc instance $\mathcal{I}$ has constant size and can therefore be solved in constant time.
We add
\begin{equation*}
N \coloneqq \max\left\{0,2^{\left\lceil\frac{|E|}{h(|X|)}\right\rceil} - |X|\right\}
\end{equation*}
 new colorless dummy points to the \gckc instance $\mathcal{I}$ to obtain a new blown-up \gckc instance $\overline{\mathcal{I}}$. These dummy points can be added to an arbitrary location as they are only used to blow up the instance size and, being colorless, do not have any further impact. Hence, the new $\gckc$ instance $\overline{\mathcal{I}}$ has size
\begin{equation}\label{eq:sizeNewInstance}
|\overline{X}| = \max\left\{ |X|, 2^{\left\lceil \frac{|E|}{h(|X|)} \right\rceil} \right\}\enspace
\end{equation}
with $\gamma$ many colors and is equivalent to the original instance $\mathcal{I}$. Notice that
\begin{equation*}
f(|\overline{X}|) \geq h(|\overline{X}|)\log |\overline{X}| \geq |E| \cdot \frac{h(|\overline{X}|)}{h(|X|)} \geq |E| = \gamma\enspace,
\end{equation*}
where the above inequalities follow by the properties of the function $h$, including that $h$ is non-decreasing, and~\eqref{eq:sizeNewInstance}. Hence, algorithm $\mathcal{A}$ is applicable to $\overline{\mathcal{I}}$ and, because $\overline{\mathcal{I}}$ and $\mathcal{I}$ are equivalent instances, $\mathcal{A}$ solves the original $\VCthree(G,t)$ instance we started with. It remains to show that its running time violates the Exponential Time Hypothesis.

The running time to solve $\overline{\mathcal{I}}$ through $\mathcal{A}$ is upper bounded by
\begin{equation*}
c_1 |\overline{X}|^{c_2} = c_1 \max \left\{
|X|^{c_2},
2^{ c_2 \left\lceil\frac{|E|}{h(|X|)} \right\rceil}
\right\} = 2^{o(|E|)}\enspace,
\end{equation*}
where the last equality uses $h(n) = \omega(1)$ and the fact that $|E| \geq |V|-1 = |X|-1$, which holds because $G$ is connected. Hence, the total running time to solve the original $\VCthree(G,t)$ instance, including the above-mentioned reductions (for example to decompose the $\VCthree(G,t)$ instance if $G$ is not connected), is $2^{o(|E|)}\poly(|V|)$.
We conclude by showing that this indeed contradicts the Exponential Time Hypothesis. Observe that for any instance of $\VCthree(G,t)$, any feasible vertex cover has size at least $|E|/3$, and so, if the given $t$ is less than $|E|/3$, then we can easily decide that the instance has no vertex cover of size $t$. Thus, any non-trivial instance of $\VCthree(G,t)$ satisfies $t \geq |E|/3$. This means that we can obtain an algorithm that solves $\VCthree(G,t)$ in time $2^{o(t)}\poly(|V|)$, which contradicts part~\ref{thm: hardness-i} of Theorem~\ref{thm:VC-hardness}.

Finally, we highlight that the function $h(n)$ does not need to be known or computed explicitly to perform the reduction. By our choice of $N$, the blown-up \gckc instance $\overline{\mathcal{I}}$ has a size $|\overline{X}|$ that is either $|X|=|V|$ or a power of two between $|X|$ and $2^{|E|}$. Hence, one could simply run $\mathcal{A}$ in parallel for each of these linearly many options of the size of the blown-up instance and terminate as soon as the first one of these parallel computations terminates.    
\end{proof}

\subsection{Hardness for bi-criteria algorithms}\label{subsec:bicriteria}

In this section, we extend the hardness result stated in Theorem~\ref{thm:mainHardness} to bi-criteria algorithms. We start by observing that the reduction described in Lemma~\ref{lemma:reduction} can be carried out even if we start with a general Set Cover instance. We now define the Set Cover problem.
\begin{definition}
Let $U = \{u_1, \ldots, u_n\}$ be a universe of $n$ elements, and let $\mathcal{S} = \{S_1, \ldots, S_m\}$ be a family of subsets of $U$, i.e., $S_i \subseteq U$ for every $i \in [m]$. The Set Cover problem asks to compute the smallest set $J \subseteq [m]$ such that $\bigcup_{j \in J} S_j = U$.
\end{definition}

Set Cover is a well-understood $\NP$-hard problem. We are interested in the hardness of Set Cover, which, after a long series of works, was settled by Dinur and Steurer~\cite{DBLP:conf/stoc/DinurS14}; we state their result as Theorem~\ref{thm:set-cover-hardness}. We note that since we are not interested in optimizing the constant that appears in the main theorem of this section, any known $\Omega(\log n)$-hardness result for Set Cover suffices to get Theorem~\ref{thm:bicriteria-hardness}, proved below.
\begin{theorem}[\cite{DBLP:conf/stoc/DinurS14}]\label{thm:set-cover-hardness}
For every $\varepsilon > 0$, it is $\NP$-hard to approximate Set Cover for instances with universe size $n$ and $m \leq \poly(n)$ sets to within a factor of $(1 - \varepsilon) \ln n$.
\end{theorem}

Mimicking the proof of Lemma~\ref{lemma:reduction}, we can prove the following lemma.
\begin{lemma}\label{lemma:reduction-set-cover}
Given a Set Cover instance $(U,\mathcal{S})$ and an integer $t \in [|\mathcal{S}|]$, there exists a polynomial-time algorithm that constructs a \gckc instance on the real line with the following properties:
\begin{enumerate}
    \item $|X| = |\mathcal{S}|$, $\g = |U|$, and $k = t$,
    \item it has a feasible solution of radius $0$ if and only if there exists a set cover in the given instance of size at most $t$.
\end{enumerate}
\end{lemma}
\begin{proof}
We construct a \gckc instance with $|X|=|\mathcal{S}|$, $\g = |U|$, and $k = t$ as follows. Let $U = \{u_1,\ldots, u_{|U|}\}$ and $\mathcal{S} = \{S_1, \ldots, S_{|\mathcal{S}|}\}$. We set $X = \{1, \ldots, |\mathcal{S}|\} \subseteq \mathbb{R}$. Each element $u_\ell \in U$ corresponds to a distinct color $X_\ell = \{i \in [|\mathcal{S}|]: u_\ell \in S_i\}$. Note that the number of colors of a point $i \in X$ is equal to $|S_i|$. We also set the covering requirement for each color to be $m_\ell = 1$. It is easy to observe that the given Set Cover instance is a YES instance if and only if there is a solution to our $\gckc$ instance of radius $0$. This implies the result.    
\end{proof}

Putting together Theorem~\ref{thm:set-cover-hardness} and Lemma~\ref{lemma:reduction-set-cover}, we can now prove Theorem~\ref{thm:bicriteria-hardness}
\begin{proof}[Proof of Theorem~\ref{thm:bicriteria-hardness}]
Let $(U, \mathcal{S})$ be a Set Cover instance with optimal value $k^*$ and $|\mathcal{S}| \leq \poly(|U|)$. Suppose that there exists a $\left(c \log|X|, \beta\right)$ bi-criteria approximation algorithm for \gckc, for a constant $c > 0$ that will be determined later. For every $k \in \{1, \ldots, \min\{|U|, |\mathcal{S}|\}\}$, we use the reduction of Lemma~\ref{lemma:reduction-set-cover} (with $t = k$) to get a \gckc instance with $|X| = |\mathcal{S}|$, and $\gamma = |U|$. For $k = k^*$, the resulting \gckc instance, by Lemma~\ref{lemma:reduction-set-cover}, has a feasible solution of radius $0$. Thus, for this instance our algorithm will return a solution of radius at most $\beta \cdot 0 = 0$ with at most $k^* \cdot c \log|X|$ centers. It is easy to see that such a solution corresponds to a set cover solution for the given Set Cover instance of size at most $k^* \cdot c \log|\mathcal{S}|$. Since $|\mathcal{S}| \leq \poly(|U|)$, this means that the returned set cover has size at most $k^* \cdot c' \log |U|$, for some constant $c' > 0$ that depends on $c$ and the hidden universal constants in the $|\mathcal{S}| \leq \poly(|U|)$ assumption. Thus, by considering all constructed \gckc instances for which a solution of radius $0$ was returned and picking the solution of radius $0$ returned for the smallest $k$, we obtain a set cover of size at most  $k^* \cdot c' \log |U|$. By setting the constant $c$ appropriately (it is easy to see that this can always be done for sufficiently small $c$), this now contradicts Theorem~\ref{thm:set-cover-hardness}. We conclude that it is $\NP$-hard to decide whether a \gckc instance has a solution of radius $0$, even if we allow solutions that open up to $c \log|X| \cdot k$ centers.
\end{proof}

\section{Conclusion}

In this work, we presented a technique for obtaining true constant-factor approximation algorithms for $k$-center problems with multiple covering constraints on the points to be covered. This leads to a polynomial-time $4$-approximation algorithm for $\gamma$-Colorful $k$-Center, where $\gamma$, the number of colors, is assumed to be constant, as well as a polynomial-time $4$-approximation algorithm for the more general Fair $\gamma$-Colorful $k$-Center problem. 

We note here that our results extend to the supplier setting, where there are distinct sets of facilities and clients, and one is allowed to open $k$ facilities in order to cover clients. For such settings, we obtain a polynomial-time $5$-approximation algorithm for the Fair $\gamma$-Colorful $k$-Supplier problem. The extension of our arguments to this setting is done by using a standard technique: we first find clients $C$ that constitute a $4$-approximate solution to the corresponding Center problem and then pick a facility $f_c\in B(c,r)$ for each $c\in C$. Using the notation introduced in the description of Algorithm~\ref{alg:partition}, we note that terminating Algorithm~\ref{alg:partition} once $\max_{u\in U} x(u) = 0$ does not affect the remaining steps in our approximation algorithms. Hence we may assume that $x(s)>0$ for all $s\in S$, which guarantees the existence of a facility in $B(s,r)$. We also clarify that the ``guessing a few centers" part of our algorithm performed in Lemma~\ref{lemma:DP_lottery} can be applied directly to facilities with no issues arising.

On the negative side, we show that Colorful $k$-Center is inapproximable when the number of colors is assumed to be part of the input.

There are still some open questions remaining; we address two of them, which we find most natural and interesting:
\begin{enumerate}
    \item The currently known hardness of $\gamma$-Colorful $k$-Center is $2-\varepsilon$, inherited from the standard $k$-Center problem, while (for constant $\gamma$) we give a polynomial-time $4$-approximation, and, as already mentioned, in an independent work, Jia, Sheth, and Svensson~\cite{DBLP:conf/ipco/JiaSS20} give a polynomial-time $3$-approximation. It would be interesting to close this gap.
    \item $\gamma$-Colorful $k$-Center naturally generalizes to the knapsack and matroid versions of it, where the set of centers that are opened must satisfy a knapsack or a matroid constraint. Currently, our technique does not easily generalize to such settings, so new ideas might be needed to handle these problems.
\end{enumerate}

\appendix

\section{Technical theorems}\label{appendix:missing-proofs}

\begin{theorem}[\cite{DBLP:conf/esa/Bandyapadhyay0P19}]\label{thm:general-sparse}
Let $(X,d)$ be a finite metric space, and suppose that the following polytope
\begin{equation*}
\mathcal{T} = \left\{ (x,y)\in [0,1]^X \times [0,1]^X \: \middle| \:
\renewcommand\arraystretch{1.2}
\begin{array}{>{\displaystyle}rcl@{\quad}l}
    \sum_{v\in X}y(v)                &\leq &k & \\
    \sum_{v\in B(u,r)} y(v)   &\geq&x(u) & \forall u\in X \\ 
    \sum_{u\in X} a_\ell(u) x(u) &\geq &b_{\ell} &\forall \ell \in [t]  
\end{array}
\right\}
\end{equation*}
is not empty, where $k \in \{1, \ldots, |X|\}$, $t \in \mathbb{Z}_{\geq 0}$, $a_\ell \in \mathbb{R}_{\geq 0}^X$ and $b_{\ell} \in \mathbb{R}_{\geq 0}$ for every $\ell\in[t]$, and $r \geq 0$. Let $(x,y) \in \mathcal{T}$, and let $(S,\mathcal{D})$ be a partition obtained by running Algorithm~\ref{alg:partition} with input $(x,y)$. Then, if $y(B(S, r)) \leq k - t + 1$, we can find in polynomial time a set $C \subseteq X$ with $|C| \leq k$ satisfying $a_\ell(B(C,4r)) \geq b_\ell$ for all $\ell \in [t]$.
\end{theorem}
\begin{proof}
Let $S = \{s_1, \ldots, s_q\}$ and $\mathcal{D} = \{D_1, \ldots, D_q\}$ be the partition obtained by running Algorithm~\ref{alg:partition} with input $(x,y)$. It is easy to see that $(S,\mathcal{D})$ satisfies all three properties of an $(x,y)$-good partition, and so, by slightly abusing terminology, we will call it an $(x,y)$-good partition.\footnote{Note that the only reason why this is a slight abuse of terminology is because we defined $(x,y)$-good partitions only for points in $\mathcal{P}$. Moreover, contrary to $\mathcal{T}$, the decription of the polytope $\mathcal{P}$ contains specific constraints for the covering requirements of the colors. However, these constraints did not play any role in showing that Algorithm~\ref{alg:partition} returns an $(x,y)$-good partition (see proof of Lemma~\ref{thm:partition-algorithm}).}
We now assume that $q \geq k + 1$, since otherwise, the set of centers $S$ is already a feasible solution, as $B(S,4r) = X$. We claim that the simplified LP given below is feasible and has optimal value  at most $y(B(S,r))$. 
\begin{equation}\label{eq:sparseLP}
\renewcommand\arraystretch{1.5}
\begin{array}{r>{\displaystyle}rll@{\quad}l}
    \min &  \sum\limits_{i=1}^q z_i \hspace*{34pt} &  &\\
    \textrm{s.t.} & \sum\limits_{i=1}^q a_\ell(D_i) \cdot z_i  &\geq&  b_\ell  &\forall \ell \in [t] \\
                  &z_i &\in& [0,1]  &\forall i\in[q] \enspace.
\end{array}
\end{equation}
This is indeed the case because we can construct a feasible point to the above LP with objective value at most $y(B(S,r))$ as follows. Let $z_i = \min\{1, y(B(s_i,r))\}$ for all $i\in [q]$. Because $(\mathcal{D}, S)$ is a $(x,y)$-good partition, property~\ref{item:goodClusterColors} of Definition~\ref{def:good-partition} implies that 
\begin{equation*}
\sum_{i\in [q]} a_\ell(D_i) z_i  \geq \sum_{i\in [q]} \sum_{u \in D_i}a_\ell(u) x(u) =\sum_{u\in X} a_\ell(u) x(u)\geq  b_\ell \qquad \forall \ell\in [t]\enspace,
\end{equation*}
(here we also use the fact that $x(u) \leq 1$ for all $u \in X$, as $(x,y) \in \mathcal{T}$), i.e., $z$ is a feasible solution of the above LP, and its objective value is $\sum_{i \in [q]} z_i \leq y(B(S, r))$.

Suppose now that the hypothesis holds, i.e., $y(B(S,r)) \leq k - t + 1$. In particular, this means that $t \leq k + 1 \leq q$. Note that if $t = k + 1$, then $y(B(S,r)) = 0$, which, by the greediness of Algorithm~\ref{alg:partition}, further implies that $b_\ell = 0$ for every $\ell \in [t]$. Such a case is trivial, as we can simply set $C \coloneqq \emptyset$. Thus, from now on, we assume that $t \leq k < q$. By the above discussion, the optimal value of the above simplified LP is at most $k - t + 1$. We consider an optimal extreme point solution $z^*$ of LP~\eqref{eq:sparseLP}. A standard sparsity argument implies that $z^*$ has at most $t$ fractional variables. Indeed, $z^*$ is defined by $q$ linearly independent and tight constraints of~\eqref{eq:sparseLP}, among which at most $t$ many are not of type $z_i\geq 0$ or $z_i \leq 1$. Hence, this implies that there are at least $q-t$ $z^*$-tight constraints of~\eqref{eq:sparseLP} of type $z_i^* =0$ or $z_i^*=1$. This in turn implies that $z^*$ has at most $t$ fractional components.

Furthermore, the number of strictly positive components of $z^*$ is at most $k$. To see this, note that if $k-t+1$ components of $z^*$ are equal to $1$, all other entries must be $0$ because $z^*$ is an optimal solution to~\eqref{eq:sparseLP}, which has objective value no more than $k-t+1$. Otherwise, there are at most $k-t$ variables that are equal to $1$ and, together with at most $t$ fractional variables, there are at most $k$ strictly positive entries. Therefore, the set of centers $C=\{s_i \in X \ | \ z^*_i>0\} \subseteq S$ has size at most $k$ and satisfies $a_\ell(B(C,4r)) \geq b_\ell$ for all $\ell \in [t]$, because $\bigcup_{c \in C} B(c,4r) \supseteq \bigcup_{i: z^*_i>0} D_i$, as $D_i \subseteq B(s_i, 4r)$ for all $i\in [q]$.   
\end{proof}

For completeness, we now discuss how the dynamic programming problems appearing in our approaches can be solved in the claimed running time.
\begin{theorem}\label{thm:dp}
Consider the following binary program:
\begin{equation*}
\begingroup
\renewcommand*{\arraystretch}{1.5}
\begin{array}{r>{\displaystyle}rll@{\quad}l}
\max& \sum_{i = 1}^q w(i) \cdot z(i)&\\
 \mathrm{s.t.}& \sum_{i = 1}^q a_\ell(i)\cdot z(i) &\geq& m_\ell  &\forall \ell\in [\gamma]\\
&\sum_{i = 1}^q z(i)  &\leq&   \kappa \\
&z &\in& \{0,1\}^{q} \enspace,
\end{array}
\endgroup
\end{equation*}
where $\gamma\in \mathbb{Z}_{\geq 1}$, $w \in \mathbb{R}_{\geq 0}^q   $, $a_\ell \in \{0, \ldots, M\}^q$ and $m_\ell \in \{0, \ldots, M\}$ for all $\ell \in [\gamma]$, where $M$ is some positive integer number, and $\kappa \in [q]$. Then, the above program can be solved in time $O(\gamma q^2 M^\gamma)$.
\end{theorem}
\begin{proof}
The above binary program can be solved using standard dynamic programming techniques. More precisely, we define the following DP table. For every $i \in \{0, \ldots, q\}$, $M_\ell \in \{0, \ldots m_\ell\}$ for every $\ell \in [\gamma]$, and $j \in \{0, \ldots, \kappa\}$, let $A[i, M_1, \ldots, M_\gamma, j]$ be the maximum objective value of any vector $z\in \{0,1\}^q$ that satisfies
\begin{enumerate}
\item $\{t \in [q]: \; z(t) = 1\} \subseteq [i]$,
\item $\sum_{t = 1}^q z(t) \leq j$, and
\item $\sum_{t = 1}^q a_\ell(t)\cdot z(t) \geq M_\ell$  for every $\ell \in [\gamma]$.
\end{enumerate}
Initialization is easy to define. For all non-trivial tuples $[i, M_1, \ldots, M_\gamma, j]$, by setting $b_\ell \coloneqq \max\{M_\ell - a_\ell(i), 0\}$ for every $\ell \in [\gamma]$, we have
\begin{align*}
    A[i, M_1, \ldots, M_\gamma, j] = \max\{&w(i) + A[i-1, b_1, \ldots, b_\gamma, j-1], A[i-1, M_1, \ldots, M_\gamma, j] \}\enspace.
\end{align*}
There are $O(q \kappa M^\gamma)$ table entries in total, and each entry can be computed in time $O(\gamma)$. Thus, the DP can be solved in time $O(\gamma q^2  M^\gamma)$.   
\end{proof}
We remark that the $O(\gamma)$ update time per table entry in the above proof can be reduced to $O(1)$ amortized update time per table entry through a more careful analysis. However, the resulting slight reduction in running time from $O(\gamma q^2 M^{\gamma})$ to $O(q^2 M^{\gamma})$ is irrelevant for our purposes.

\section{A limiting example for the framework of Chakrabarty and Negahbani~\cite{DBLP:journals/talg/ChakrabartyN19}\label{appendix:bad-example-deeparnab}}

A natural way to extend the approach of~\cite{DBLP:journals/talg/ChakrabartyN19} is the following procedure. Given a point $(x,y)\in \mathbb{R}^X\times \mathbb{R}^X$, we first run Algorithm~\ref{alg:partition} (with balls of radius $2$ at each step) to get a partition of $X$, and then we use dynamic programming to decide whether it is possible to select at most $k$ clusters of this partition so that the covering requirements for all colors are satisfied. Such a selection, if it exists, gives a $2$-approximation. If there is no such selection, we want to return a hyperplane separating $(x,y)$ from $\ip$, as in~\cite{DBLP:journals/talg/ChakrabartyN19}.

However, there is an instance and a point $(x,y)$, given below, such that neither the partition will lead to a solution nor is it possible to separate $(x,y)$ from $\ip$. Thus any such procedure needs to deal with this limitation.

In Figure~\ref{fig:deeparnab_counterexample}, we present an instance of $\g$-Colorful $k$-Center with $\g=k=2$ in the one-dimensional Euclidean space; hence $X\subseteq \mathbb{R}$. There are two colors, red and blue; the red points are represented as red circles and the blue points as blue squares. The color covering requirements are $m_1=m_2=3$. It is easy to see that there are no integral solutions of radius $0$, hence any solution with radius $1$ is optimal. We consider two different optimal solutions:
\begin{itemize}
    \item  $C_1 = \{1, M + 1\}$ with corresponding clustering $\mathcal{C}_1 = \left\{ \{1,2\}, \{M+1, M+2 \} \right\}$,
    \item  $C_2 = \{4, M+4\}$ with corresponding clustering $\mathcal{C}_2 = \left\{ \{3,4\}, \{M+3, M+4\} \right\}$.
\end{itemize}
We clarify that in the above, we slightly abuse notation; if there are multiple points in a location, we only pick one of them as a center, while in the corresponding clustering, all points in a covered location participate in the clustering. It is easy to verify that the above clusterings are indeed feasible solutions of radius 1, and thus, they are optimal solutions.

We now define the fractional solution $(x,y) \in \mathbb{R}^X \times \mathbb{R}^X$, where $x=\frac 1 2 \left( \chi^{\mathcal{C}_1}+\chi^{\mathcal{C}_2}\right)$ and $y=\frac 1 2 \left(\chi^{C_1} + \chi^{C_2}\right)$. Observe that we have $x(u)=\frac 1 2$ for all $u\in X$. 

In the above example, given the defined point $(x,y)$ as input, Algorithm~\ref{alg:partition} may return the indicated partitioning $\{D_1, D_2, D_3, D_4\}$. We stress here that there are ties, and in order to get this partitioning we resolve them adversarially. Note that there is no specified way to resolve such ties in~\cite{DBLP:journals/talg/ChakrabartyN19} and it seems highly unclear how to design a procedure that always break ties in a good way even if there is a good way to break them. Observe now that no combination of two of these resulting clusters satisfies the covering requirement, so the partitioning does not lead to a solution. However, we cannot possibly find an appropriate separating hyperplane because $(x,y)\in \ip$ by construction.

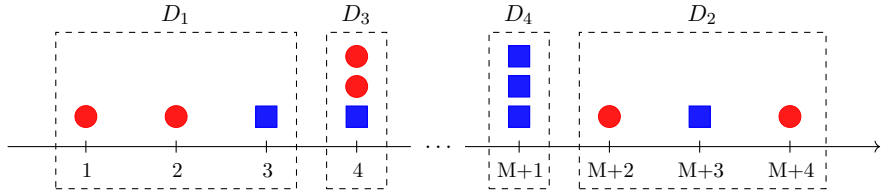
\begin{figure}[ht]
\begin{center}

    \scalebox{0.8}{
	\begin{tikzpicture}
	\pgfmathsetmacro{\myheight}{2.5}
	\pgfmathsetmacro{\myshift}{0.8}

	\draw[-, line width = 0.5pt, black] (0.5,1) -- (7.2,1) {{}};
	\draw[->, line width = 0.5pt, black] (8.1,1) -- (15,1) {{}};
	
	\node[right] at (7.3,0.99) {$\ldots$};
		
	\foreach \i in {1,...,4}
	{
		\draw[-,line width = 0.5pt, black] (0.3 + 1.5*\i,1.1) -- (.3 + 1.5*\i,0.9)  {{}};
		\node[right] at (0.1 + 1.5*\i,0.6) {\small  \i};
	}

	\foreach \i in {1,...,4}
	{
		\draw[-,line width = 0.5pt, black] (7.5+1.5*\i,1.1) -- (7.5+1.5*\i,0.9)  {{}};
		\node[right] at (7+1.5*\i,0.6) {\small  M+\i};
	}

	\node[circle,draw=red,fill=red!90,minimum size=10]  at (1.8,1.5) {{}};
	\node[circle,draw=red,fill=red!90,minimum size=10]  at (3.3,1.5) {{}};
	\node[draw=blue,fill=blue!90,minimum size=10]  at (4.8,1.5) {{}};
	\node[draw=blue,fill=blue!90,minimum size=10]  at (6.3,1.5) {{}};
	\node[circle,draw=red,fill=red!90,minimum size=10]  at (6.3,2) {{}};
	\node[circle,draw=red,fill=red!90,minimum size=10]  at (6.3,2.5) {{}};	
	
	\node[draw=blue,fill=blue!90,minimum size=10]  at (9,1.5) {{}};
	\node[draw=blue,fill=blue!90,minimum size=10]  at (9,2) {{}};
	\node[draw=blue,fill=blue!90,minimum size=10]  at (9,2.5) {{}};
	\node[circle,draw=red,fill=red!90,minimum size=10]  at (10.5,1.5) {{}};
	\node[draw=blue,fill=blue!90,minimum size=10]  at (12,1.5) {{}};
	\node[circle,draw=red,fill=red!90,minimum size=10]  at (13.5,1.5) {{}};

	\draw [black!100,line width = 0.5pt, dashed]  (1.3,0.3) -- (5.3,0.3) -- (5.3,2.9) --  node [text width=0.5cm, black, midway, above ]{$D_1$} (1.3,2.9) -- cycle;
	
	\draw [black!100,line width = 0.5pt, dashed]  (5.8,0.3) -- (6.8,0.3) -- (6.8,2.9) -- node [text width=0.5cm, black, midway, above ]{$D_3$} (5.8,2.9) -- cycle;

	\draw [black!100,line width = 0.5pt,dashed]  (8.5,0.3) -- (9.5,0.3) -- (9.5,2.9) -- node [text width=0.5cm, black, midway, above ]{$D_4$} (8.5,2.9) -- cycle;

	\draw [black!100,line width = 0.5pt,dashed]  (10,0.3) -- (14.1,0.3) -- (14.1,2.9) -- node [text width=0.5cm, black, midway, above ]{$D_2$} (10,2.9) -- cycle;

	\end{tikzpicture} 	    }
	\caption{A limiting example for the framework of Chakrabarty and Negahbani~\cite{DBLP:journals/talg/ChakrabartyN19}.}
	\label{fig:deeparnab_counterexample}
	\end{center}
\end{figure}


\begin{thebibliography}{10}
\providecommand{\url}[1]{\texttt{#1}}
\providecommand{\urlprefix}{URL }
\providecommand{\doi}[1]{https://doi.org/#1}

\bibitem{an_2017_lp-based}
An, H.C., Singh, M., Svensson, O.: {LP}-based algorithms for capacitated
  facility location. SIAM Journal on Computing  \textbf{46}(1),  272--306
  (2017)

\bibitem{DBLP:conf/icml/BackursIOSVW19}
Backurs, A., Indyk, P., Onak, K., Schieber, B., Vakilian, A., Wagner, T.:
  Scalable fair clustering. In: Proceedings of the 36th International
  Conference on Machine Learning ({ICML}). pp. 405--413 (2019)

\bibitem{DBLP:conf/esa/Bandyapadhyay0P19}
Bandyapadhyay, S., Inamdar, T., Pai, S., Varadarajan, K.R.: A constant
  approximation for {C}olorful $k$-{C}enter. In: Proceedings of the 27th Annual
  European Symposium on Algorithms {(ESA)}. pp. 12:1--12:14 (2019)

\bibitem{DBLP:conf/nips/BeraCFN19}
Bera, S.K., Chakrabarty, D., Flores, N., Negahbani, M.: Fair algorithms for
  clustering. In: Proceedings of the 33rd International Conference on Neural
  Information Processing Systems ({NeurIPS}). pp. 4955--4966 (2019)

\bibitem{DBLP:conf/approx/Bercea0KKRS019}
Bercea, I.O., Gro{\ss}, M., Khuller, S., Kumar, A., R{\"{o}}sner, C., Schmidt,
  D.R., Schmidt, M.: On the cost of essentially fair clusterings. In:
  Proceedings of the 22nd International Conference on Approximation Algorithms
  for Combinatorial Optimization Problems ({APPROX}). pp. 18:1--18:22 (2019)

\bibitem{DBLP:journals/jcss/CaiJ03}
Cai, L., Juedes, D.W.: On the existence of subexponential parameterized
  algorithms. Journal of Computer and System Sciences  \textbf{67}(4),
  789--807 (2003)

\bibitem{carr_2000_strengthening}
Carr, R.D., Fleischer, L.K., Leung, V.J., Phillips, C.A.: Strengthening
  integrality gaps for capacitated network design and covering problems. In:
  Proceedings of the 11th Annual ACM-SIAM Symposium on Discrete Algorithms
  (SODA). pp. 106--115 (2000)

\bibitem{DBLP:conf/icalp/ChakrabartyGK16}
Chakrabarty, D., Goyal, P., Krishnaswamy, R.: The non-uniform $k$-center
  problem. In: Proceedings of the 43rd International Colloquium on Automata,
  Languages, and Programming {(ICALP)}. pp. 67:1--67:15 (2016)

\bibitem{DBLP:journals/talg/ChakrabartyN19}
Chakrabarty, D., Negahbani, M.: Generalized center problems with outliers.
  {ACM} Transactions on Algorithms  \textbf{15}(3),  41:1--41:14 (2019)

\bibitem{DBLP:conf/soda/CharikarKMN01}
Charikar, M., Khuller, S., Mount, D.M., Narasimhan, G.: Algorithms for facility
  location problems with outliers. In: Proceedings of the 12th Annual Symposium
  on Discrete Algorithms ({SODA}). pp. 642--651 (2001)

\bibitem{DBLP:journals/algorithmica/ChenLLW16}
Chen, D.Z., Li, J., Liang, H., Wang, H.: Matroid and knapsack center problems.
  Algorithmica  \textbf{75}(1),  27--52 (2016)

\bibitem{DBLP:conf/nips/Chierichetti0LV17}
Chierichetti, F., Kumar, R., Lattanzi, S., Vassilvitskii, S.: Fair clustering
  through fairlets. In: Proceedings of the 31st International Conference on
  Neural Information Processing Systems (NIPS). pp. 5029--5037 (2017)

\bibitem{DBLP:conf/stoc/DinurS14}
Dinur, I., Steurer, D.: Analytical approach to parallel repetition. In:
  Proceedings of the 46th Annual Symposium on the Theory of Computing {(STOC)}.
  pp. 624--633 (2014)

\bibitem{Garey:1974:SNP:800119.803884}
Garey, M.R., Johnson, D.S., Stockmeyer, L.: Some simplified {NP}-complete
  problems. In: Proceedings of the 6th Annual ACM Symposium on Theory of
  Computing ({STOC}). pp. 47--63 (1974)

\bibitem{DBLP:journals/tcs/Gonzalez85}
Gonzalez, T.F.: Clustering to minimize the maximum intercluster distance.
  Theoretical Computer Science  \textbf{38},  293--306 (1985)

\bibitem{grandoni_2018_improved}
Grandoni, F., Kalaitzis, C., Zenklusen, R.: Improved approximation for tree
  augmentation: Saving by rewiring. In: Proceedings of the 50th ACM Symposium
  on Theory of Computing (STOC). pp. 632--645 (2018)

\bibitem{schrijver2012geometric}
Gr{\"o}tschel, M., Lov{\'a}sz, L., Schrijver, A.: Geometric algorithms and
  combinatorial optimization, vol.~2. Springer (2012)

\bibitem{HarrisPST19}
Harris, D.G., Pensyl, T., Srinivasan, A., Trinh, K.: A lottery model for
  center-type problems with outliers. {ACM} Transactions on Algorithms
  \textbf{15}(3),  36:1--36:25 (2019)

\bibitem{DBLP:journals/mor/HochbaumS85}
Hochbaum, D.S., Shmoys, D.B.: A best possible heuristic for the
  \emph{k}-{C}enter {P}roblem. Mathematics of Operations Research
  \textbf{10}(2),  180--184 (1985)

\bibitem{DBLP:journals/jacm/HochbaumS86}
Hochbaum, D.S., Shmoys, D.B.: A unified approach to approximation algorithms
  for bottleneck problems. Journal of the {ACM}  \textbf{33}(3),  533--550
  (1986)

\bibitem{DBLP:journals/dam/HsuN79}
Hsu, W., Nemhauser, G.L.: Easy and hard bottleneck location problems. Discrete
  Applied Mathematics  \textbf{1}(3),  209--215 (1979)

\bibitem{DBLP:conf/ipco/JiaSS20}
Jia, X., Sheth, K., Svensson, O.: Fair colorful k-center clustering. In:
  Proceedings of the 21st International Conference on Integer Programming and
  Combinatorial Optimization {(IPCO)}. pp. 209--222 (2020)

\bibitem{levi_2008_approximation}
Levi, R., Lodi, A., Sviridenko, M.: Approximation algorithms for the
  capacitated multi-item lot-sizing problem via flow-cover inequalities.
  Mathematics of Operations Research  \textbf{33}(2),  461--474 (2008)

\bibitem{li_2016_approximating}
Li, S.: Approximating {c}apacitated $k$-{m}edian with $(1+\epsilon)k$ open
  facilities. In: Proceedings of the 27th Annual ACM Symposium on Discrete
  Algorithms (SODA). pp. 786--796 (2016)

\bibitem{li_2017_uniform}
Li, S.: On uniform capacitated $k$-median beyond the natural {LP} relaxation.
  ACM Transactions on Algorithms  \textbf{13}(2),  22:1--22:18 (2017)

\bibitem{nutov_2017_tree}
Nutov, Z.: On the tree augmentation problem. In: Proceedings of the 25th Annual
  Symposium on Algorithms (ESA). pp. 61:1--61:14 (2017)

\bibitem{DBLP:conf/icalp/Rosner018}
R{\"{o}}sner, C., Schmidt, M.: Privacy preserving clustering with constraints.
  In: Proceedings of the 45th International Colloquium on Automata, Languages,
  and Programming ({ICALP}). pp. 96:1--96:14 (2018)

\end{thebibliography}
\end{document}